\documentclass[12pt]{article}
\usepackage[utf8]{inputenc}
\usepackage[letterpaper]{geometry}
\usepackage{booktabs}
\usepackage{amsmath}
\usepackage{amsthm}
\usepackage{amssymb}
\usepackage{setspace}
\usepackage[authoryear]{natbib}
\usepackage[unicode=true,pdfusetitle,
 bookmarks=true,bookmarksnumbered=false,bookmarksopen=false,
 breaklinks=false,pdfborder={0 0 1},backref=false,colorlinks=false]
 {hyperref}
\hypersetup{
 colorlinks,linkcolor={blue!50!black},citecolor={blue!50!black},urlcolor={blue}}

\makeatletter

\providecommand{\tabularnewline}{\\}

\theoremstyle{plain}
\newtheorem{assumption}{\protect\assumptionname}
\theoremstyle{definition}
\newtheorem*{example*}{\protect\examplename}
\theoremstyle{plain}
\newtheorem{thm}{\protect\theoremname}
\theoremstyle{remark}
\newtheorem{rem}{\protect\remarkname}
\theoremstyle{plain}
\newtheorem{cor}{\protect\corollaryname}
\theoremstyle{remark}
\newtheorem*{acknowledgement*}{\protect\acknowledgementname}
\theoremstyle{plain}
\newtheorem{lem}{\protect\lemmaname}

\providecommand{\tabularnewline}{\\}

\theoremstyle{plain}

\usepackage{graphicx}

\theoremstyle{plain}

\usepackage{enumitem}
\usepackage{rotfloat}
\usepackage{amsfonts}
\usepackage{graphics}
\usepackage{color}
\usepackage{latexsym}
\usepackage{mathrsfs}
\usepackage{multirow}
\usepackage{rotating}
\usepackage{threeparttable}
\usepackage{lscape}
\usepackage{xcolor}
\usepackage{bbm}

\allowdisplaybreaks

\theoremstyle{plain}

\makeatother

\providecommand{\acknowledgementname}{Acknowledgement}
\providecommand{\assumptionname}{Assumption}
\providecommand{\corollaryname}{Corollary}
\providecommand{\examplename}{Example}
\providecommand{\lemmaname}{Lemma}
\providecommand{\remarkname}{Remark}
\providecommand{\theoremname}{Theorem}

\begin{document}

\global\long\def\a{\alpha}%
\global\long\def\b{\beta}%
\global\long\def\g{\gamma}%
\global\long\def\d{\delta}%
\global\long\def\e{\epsilon}%
\global\long\def\l{\lambda}%
\global\long\def\t{\theta}%
\global\long\def\s{\sigma}%
\global\long\def\G{\Gamma}%
\global\long\def\D{\Delta}%
\global\long\def\L{\Lambda}%
\global\long\def\T{\Theta}%
\global\long\def\O{\Omega}%
\global\long\def\R{\mathbb{R}}%
\global\long\def\N{\mathbb{N}}%
\global\long\def\Q{\mathbb{Q}}%
\global\long\def\I{\mathbb{I}}%
\global\long\def\P{\mathbb{P}}%
\global\long\def\E{\mathbb{E}}%
\global\long\def\B{\mathbb{\mathbb{B}}}%
\global\long\def\S{\mathbb{\mathbb{S}}}%
\global\long\def\X{{\bf X}}%
\global\long\def\cX{\mathscr{X}}%
\global\long\def\cY{\mathscr{Y}}%
\global\long\def\cA{\mathscr{A}}%
\global\long\def\cB{\mathscr{B}}%
\global\long\def\cM{\mathscr{M}}%
\global\long\def\cN{\mathcal{N}}%
\global\long\def\cG{\mathcal{G}}%
\global\long\def\cC{\mathcal{C}}%
\global\long\def\sp{\,}%
\global\long\def\es{\emptyset}%
\global\long\def\mc#1{\mathscr{#1}}%
\global\long\def\ind{\mathbf{\mathbbm1}}%
\global\long\def\indep{\perp}%
\global\long\def\any{\forall}%
\global\long\def\ex{\exists}%
\global\long\def\p{\partial}%
\global\long\def\cd{\cdot}%
\global\long\def\Dif{\nabla}%
\global\long\def\imp{\Rightarrow}%
\global\long\def\iff{\Leftrightarrow}%
\global\long\def\up{\uparrow}%
\global\long\def\down{\downarrow}%
\global\long\def\arrow{\rightarrow}%
\global\long\def\rlarrow{\leftrightarrow}%
\global\long\def\lrarrow{\leftrightarrow}%
\global\long\def\gto{\rightarrow}%
\global\long\def\abs#1{\left|#1\right|}%
\global\long\def\norm#1{\left\Vert #1\right\Vert }%
\global\long\def\rest#1{\left.#1\right|}%
\global\long\def\bracket#1#2{\left\langle #1\middle\vert#2\right\rangle }%
\global\long\def\sandvich#1#2#3{\left\langle #1\middle\vert#2\middle\vert#3\right\rangle }%
\global\long\def\turd#1{\frac{#1}{3}}%
\global\long\def\ellipsis{\textellipsis}%
\global\long\def\sand#1{\left\lceil #1\right\vert }%
\global\long\def\wich#1{\left\vert #1\right\rfloor }%
\global\long\def\sandwich#1#2#3{\left\lceil #1\middle\vert#2\middle\vert#3\right\rfloor }%
\global\long\def\abs#1{\left|#1\right|}%
\global\long\def\norm#1{\left\Vert #1\right\Vert }%
\global\long\def\rest#1{\left.#1\right|}%
\global\long\def\inprod#1{\left\langle #1\right\rangle }%
\global\long\def\ol#1{\overline{#1}}%
\global\long\def\ul#1{\underline{#1}}%
\global\long\def\td#1{\tilde{#1}}%
\global\long\def\upto{\nearrow}%
\global\long\def\downto{\searrow}%
\global\long\def\pto{\overset{p}{\longrightarrow}}%
\global\long\def\dto{\overset{d}{\longrightarrow}}%
\global\long\def\asto{\overset{a.s.}{\longrightarrow}}%
\global\long\def\ind{\mathbbm{\bf 1}}%

\title{Using Monotonicity Restrictions to\\
 Identify Models with Partially Latent Covariates}
\author{Minji Bang\thanks{Bang: mbang@sas.upenn.edu}, Wayne Gao\thanks{Gao: waynegao@upenn.edu},
Andrew Postlewaite\thanks{Postlewaite: apostlew@upenn.edu}, and Holger
Sieg\thanks{Sieg (corresponding author): holgers@econ.upenn.edu }\medskip{}
 \\
 Department of Economics, University of Pennsylvania\thanks{The Ronald O. Perelman Center for Political Science and Economics,
133 South 36th Street, Philadelphia, PA 19104, USA.}}
\maketitle
\begin{abstract}
\noindent This paper develops a new method for identifying econometric
models with partially latent covariates. Such data structures arise
in industrial organization and labor economics settings where data
are collected using an input-based sampling strategy, e.g., if the
sampling unit is one of multiple labor input factors. We show that
the latent covariates can be nonparametrically identified, if they
are functions of a common shock satisfying some plausible monotonicity
assumptions. With the latent covariates identified, semiparametric
estimation of the outcome equation proceeds within a standard IV framework
that accounts for the endogeneity of the covariates. We illustrate
the usefulness of our method using a new application that focuses
on the production functions of pharmacies. We find that differences
in technology between chains and independent pharmacies may partially
explain the observed transformation of the industry structure.

\medskip{}

\noindent Keywords: production functions, latent variables, endogeneity,
semiparametric estimation, monotonicity
\end{abstract}
\newpage{}

\section{{\normalsize{}\label{sec:Intro}}Introduction}

This paper develops a new method for identifying econometric models
with partially latent covariates. We show that a broad class of econometric
models that play a large role in industrial organization and labor
economics can be nonparametrically identified if the partially latent
covariate variables satisfy certain monotonicity assumptions. Examples
that fall into this class of models are a variety of different production,
skill formation, and achievement functions.\footnote{Other potential applications in applied microeconomics are discussed
in the conclusions.}

It is often plausible to assume that the different inputs or covariates
are functions of a common unobserved random shock, and we consider
models in which it is natural to impose strict monotonicity in this
common shock.\footnote{Note that this assumption is commonly used, for example, in the production
function literature as discussed by \citet*{olley1996dynamics}. In
particular, this assumption does not require that inputs are ``optimally''
chosen by competitive firms and is consistent with a broad class of
strategic and non-strategic models that may describe the agents' behavior.

} The monotonicity assumption imposes functional dependencies on the
explanatory variables as pointed out in the context of production
function estimation by \citet*{ackerberg2015identification}. The
key insight of this paper is that we can leverage the functional dependence
between inputs to achieve identification within a partially latent
covariate framework. In that sense, we turn the functional dependence
problem on its head to impute the partially latent covariates. Broadly
speaking, our imputation is in the spirit of matching algorithms \citep{Rubin-73}.
In contrast to traditional matching algorithms, we propose to match
on the expected dependent variable to impute missing covariates.

The partially latent data structure that we study in this paper arises
quite naturally in certain matched employer-employee data sets which
contain information collected from individuals as well as information
collected from businesses or establishments. In the past decades economists
made enormous strides in finding and using matched employer-employee
data, which have provided a new empirical basis for the study of workplace
organization, compensation design, mobility, and production. In this
paper we focus on cross-sectional employer-employee data sets which
are commonly used in applied microeconomics and econometrics. Following
\citet*{Abowd-Kramaz-99}, an important dimension that distinguishes
matched employer-employee data sets is the sampling design. Some sampling
designs focus on the firm, while others use the employee as the primary
unit of analysis. In this paper we focus on the later type of sampling
designs that often only sample one employee or a small random set
of employees of the firm.

This lack of sampling of all employees of the firm may not be important
from a perspective of labor economics, which largely focuses on job
mobility and wage determination. However, it is more problematic if
the focus is on productivity measurement within a production function
framework. If the survey does not sample all employees in a firm then
some important inputs are latent from the perspective of the econometrician.
As a consequence, we call this type of sampling design an ``input-based
sampling''\ approach, since the sampling unit is one of the multiple
labor input factors. We provide a formal definition of this data structure
in this paper. The main application that we study considers a production
team in which team members perform different tasks. In our data set
only one member from each team is interviewed to provide the data.
It is plausible that this person knows the team's output, but does
not have complete information about the other team members' input
choices. By randomly sampling the teams we elicit information from
all different types of team members and hence input factors.

Once we have identified the latent inputs, the estimation of the outcome
or production function can proceed using standard semiparametric methods
developed in the econometric literature. One key issue here is that
the common shock creates an endogeneity problem.\footnote{In the context of production function estimation this endogeneity
problem is referred to as the transmission bias problem since inputs
are correlated with unobserved productivity shocks \citep{marschak1944random}.} We show that we can combine our identification results with a variety
of linear, nonlinear, and semiparametric estimation strategies. In
that sense our approach is flexible and allows researchers to make
appropriate functional form assumptions if necessary. We consider
the scenario in which researchers only have access to a single cross-section
of data and rely on instrumental variables for estimation.\footnote{Hence we cannot address this endogeneity problem using panel data
with fixed effects, first advocated by \citet*{hoch1955estimation,hoch1962estimation}
and \citet*{mundlak1961empirical,mundlak1963specification}. We can
also not use more sophisticated timing assumptions within a control
function or IV frameworks as discussed, for example, in \citet*{olley1996dynamics}
and \citet*{blundell1998initial,blundell2000gmm}, \citet*{levinsohn2003estimating},
and \citet*{ackerberg2015identification}. We discuss the extension
of our methods to this scenario in the conclusions.} For example, production function estimation relies on the assumption
that differences in local input prices give rise to differences in
input choices that are uncorrelated with productivity shocks at the
local level.\footnote{Hence, local input prices can serve as valid instruments for endogenous
input choices. See \citet{griliches1998} for a critical discussion
of the assumption that these input prices are exogenous. Similarly,
skill formation and achievement function estimation requires the choice
of suitable instruments for parental inputs. For a more general discussion
of the issues encountered in estimating achievement and skill formation
functions see, among others, \citet*{Todd-Wolpin-03} and \citet*{Cunha-Heckman-Schennach-10}.}

Estimation proceeds in two steps. In finite samples, we first nonparametrically
estimate the latent input functions. Plugging the estimators into
our outcome equation, we can estimate the parameters of this function
using a standard IV estimator based on the observed and imputed inputs.
The second econometric challenge then arises for the need to account
for the sequential nature of the estimator when deriving the correct
rate of convergence and computing asymptotic standard errors. To illustrate
this we consider the standard log-linear, Cobb-Douglas model. We propose
two different estimators and provide both high-level and lower-level
conditions under which these semiparametric two-step estimators are
consistent and asymptotically normal at the usual parametric rate
of convergence. The technical proofs are based on the general econometric
theory on semiparametric two-step estimation as in \citet{newey1994asymptotic},
\citet*{newey1994large}, and \citet*{chen2003estimation}. Finally,
we show that using the conditional expectation of outcomes as the
dependent variable produces efficiency gains relative to the more
traditional estimator that uses the observed output instead.

To evaluate the performance of our estimator we conduct a variety
of Monte Carlo experiments. Our findings suggest that our estimators
are well-behaved in samples that are similar in size to those observed
in our applications discussed below. We also study the behavior of
our estimator when we pool observations across markets as is often
necessary for many practical applications.

The type of sampling design that gives rise to a partial latent data
structure arises in administrative employer-employee data sets that
are designed by statistical agencies. It is also common among profession-based
surveys that focus on one narrowly-defined occupation. These surveys
tend to be national in scope since professions have generally a national
market, the characteristics of which the survey organizers want to
know. In our main application, we use data from the National Pharmacist
Workforce Survey (NPWS) in 2000 which not only collects data for each
pharmacist that is surveyed but also a limited amount of information
at the store level including output.

We implement our new estimator using the NPWS and study differences
in the marginal product of labor inputs in pharmacies. \citet*{goldin2016most} have forcefully
argued that this is one of the most egalitarian and family-friendly
professions in which females face little discrimination in the workforce.
One potential explanation of this fact has been related to the rise
of chains that have replaced independent pharmacies in many local
markets. Here we estimate a team production function that distinguishes
between managerial and non-managerial certified pharmacists. We can,
therefore, test the hypothesis of whether managers have higher marginal products
in chains than in independent pharmacies. We find that we can reject
the null hypothesis that independent pharmacies and chains have the
same technology. Estimates for independent pharmacies are somewhat
noisy but do not suggest that there is a large difference between
managers and regular employees. Estimates for chains suggest that
managers have higher marginal products than regular employees. We thus conclude
that chains seem to improve the effectiveness of managers which may
partially explain why they have become the dominant firm type in this
industry.

This paper relates to the line of literature on production function
estimation by proposing a method to handle the problem of partially
latent inputs. Our identification strategy is based on strict monotonicity
and the consequent invertibility in a scalar unobservable, a feature
also leveraged by \citet*{olley1996dynamics} and \citet*{levinsohn2003estimating}.
They essentially use an auxiliary variable together with an input
to control for the unobserved productivity shock: investment with
capital in \citet*{olley1996dynamics} and intermediate inputs with
capital in \citet*{levinsohn2003estimating}. In comparison, we use
the output with the observed input to pin down the productivity shock.
We emphasize that the feature of functional dependence between input
variables, which was pointed out by \citet*{ackerberg2015identification}
as an underlying problem in \citet*{olley1996dynamics} and \citet*{levinsohn2003estimating},
in fact, forms the basis of our imputation strategy. While most of
these papers focus on value-added production functions, there is also
much interest in estimating gross output production functions. \citet*{Doraszelski-Jaumandreu-13}
propose a solution to the transmission bias problem that also relies
on observed firm-level variation in prices. In particular, they show
that by explicitly imposing the parameter restrictions between the
production function and the demand for a flexible input and by using
this price variation, they can recover the gross output production
function. \citet*{Gandhi-Navarro-Rivers-20} provide an alternative
identification strategy to estimate gross output production functions
that works well in short panels. Beyond these conceptual linkages,
our paper has a different focus from these papers cited above: they
focus more on the dynamic nature of capital inputs, while we focus
on the problem of partially latent inputs. Moreover, the estimation
of production functions is just one of many applications of our general
identification result. This paper shows that our methods may be even
more useful for applications outside of IO where these data structures
are more prevalent as we discuss below.

Also, we should point out that this paper is both conceptually and
technically different from previous work on missing data in linear
regression and, more generally, GMM estimation settings, such as \citet*{rubin1976inference},
\citet*{little1992regression}, \citet*{robins1994estimation}, \citet*{wooldridge2007inverse},
\citet*{graham2011efficiency}, \citet*{chaudhuri2016gmm}, \citet*{abrevaya2017gmm}
and \citet*{mcdonough2017missing}. This line of literature usually
exploits two types of conditions: first, observations with no missing
data occur with positive probability, and second, data are ``missing
at random'' (potentially with conditioning). Neither condition is
satisfied in our setting: every observation contains missing data,
and missing can be correlated with other observables as well as the
unobserved productivity shock. Instead, we rely on monotonicity in
a scalar unobservable shock to identify and impute the latent input.

Similarly, our monotonicity conditions also differentiate our paper
from the econometric literature on data combination as surveyed by
\citet*{ridder2007econometrics}, which mostly involves conditional
independence assumptions. In particular, our paper shares a similar
flavor with the ``moment-matching'' approach in this literature,
such as \citet*{angrist1992effect}, \citet*{ridder2007econometrics},
and \citet*{buchinsky2022estimation}: these papers utilize conditional
independence to achieve ``moment matching'', while we exploit monotonicity
to do so. Hence, our proposed method may also be useful as a complementary
data combination method for scenarios where our monotonicity conditions
are interpretable and justifiable.

Broadly speaking, our imputation is in the spirit of matching algorithms
\citep{Rubin-73}. In contrast to traditional matching algorithms,
we propose to match on the expected dependent variable to impute missing
covariates. Hence, we do not apply the matching approach within the
standard potential outcome framework of program evaluation based on
the potential outcome model developed by \citet*{Fisher-35}.\footnote{For a discussion of the properties of matching estimators in that
context see, among others, \citet*{Rosenbaum-Rubin-83}, \citet*{Heckman-Ichimura-Smith-Todd-98},
and \citet*{Abadie-Imbens-06}.}

The rest of the paper is organized as follows. Section \ref{sec:Identification}
presents our main identification result. Section \ref{sec:Estimation}
discusses the problems associated with estimation. Section \ref{sec:MonteCarlo}
reports the results from a Monte Carlo Study. Section \ref{sec:FirstAppl}
introduces our application focusing on the production functions of
pharmacies. It discusses our data sources and presents our main empirical
findings. Section \ref{sec:Conclusion} provides a discussion of other
potential applications and presents our conclusions.

\section{\label{sec:Identification} Identification of Partially Latent Covariates}

\subsection{\label{subsec:ModelSetup}Model Setup}

Consider the following cross-sectional econometric model:
\begin{equation}
y_{i}=F\left(x_{i1},x_{i2},u_{i}\right)+\e_{i}\label{eq:Gen_model}
\end{equation}
where $i=1,...,N$ indexes a generic observation from a \textit{random
sample}, $y_{i}$ denotes an observable scalar-valued output variable,
and $x_{i}:=\left(x_{i1},x_{i2}\right)$ denotes a two-dimensional
vector of covariates.\footnote{See Corollary \ref{cor:more_inputs} for the extension of our identification
method to settings with covariates of higher dimensions.} Both $u_{i}$ and $\e_{i}$ are scalar-valued unobserved errors,
with $u_{i}$ taken to be a structural error (such as a productivity
shock) that is endogenous with respect to $x_{i}$, while $\e_{i}$
is a ``measurement error'' that is assumed to be exogenous. The
unknown outcome function $F$ may be either parametric or nonparametric.

First, we need to define what we mean by \textit{partially latent
covariates}, a key data structure that we seek to handle in this paper.
\begin{assumption}[Partially Latent Covariates]
\label{assu:OneLatent} For each observation $i$, the econometrician
either observes $x_{i1}$ or $x_{i2}$, but never both.
\end{assumption}
Essentially, one of the two covariates $\left(x_{i1},x_{i2}\right)$
is latent in each observation in the data. In the following, it will
be convenient to write
\[
d_{i}:=\begin{cases}
1, & \text{ if }x_{i1}\text{ is observed and }x_{i2}\text{ is latent},\\
2, & \text{ if }x_{i2}\text{ is observed and }x_{i1}\text{ is latent},
\end{cases}
\]
so that effectively $\left(d_{i},\left(2-d_{i}\right)x_{i1},\left(d_{i}-1\right)x_{i2}\right)$
is observed for $i$.
\begin{example*}[Team Production Functions]
Such data structures, for example, arise when the data is collected
at the individual input level while we are interested in some firm
or team level output variable that also depends on other individual
inputs who are not surveyed in the data. Our main application focuses
on identifying and estimating team production functions.\footnote{We use the term ``team production function'' since we largely focus
on different types of labor inputs and abstract from capital or other
inputs that may be subject to dynamics and adjustment costs.} For simplicity, let us assume a log-linear Cobb-Douglas specification:
\begin{equation}
y_{i}=\a_{0}+\a_{1}x_{i1}+\a_{2}x_{i2}+u_{i}+\e_{i},\label{eq:Model_CD}
\end{equation}
where $y_{i}$ is the logarithm of the team's output, $x_{i1}$ is
the logarithm of hours worked by the first team member (a manager),
and $x_{i2}$ is the logarithm of hours worked by the second team
member (an employee).\footnote{The team production concept is also related to the concept of task
production functions, which are surveyed by \citet*{Acemoglu-Autor-11}.
\citet*{Haanwinckel-18} estimates a task production function in which
each team member specializes in a single task.} The data structure described in Assumption \ref{assu:OneLatent}
arises if the researcher interviews only one member, and not both
members of the team. We also refer to this technique as an ``\textit{input-based
sampling}'' approach. It is plausible that the interviewed team member
knows the team's output, but does not have complete information about
the other team member's input choices. Hence, the surveyed person
provides the output level, $y_{i}$, and her own hours worked, $x_{i1}$
\textit{or} $x_{i2}$, leading to the problem of partially latent
inputs as defined in Assumption \ref{assu:OneLatent}.
\end{example*}
The next assumption imposes a monotonicity condition on the outcome
function.
\begin{assumption}[Monotonicity of the Outcome Function]
 \label{assu:MonoF} $F$ is nondecreasing in all its arguments and
is strictly increasing in at least one of its arguments.
\end{assumption}
This assumption essentially states that the inputs $\left(x_{i1},x_{i2}\right)$
and the productivity shock $u_{i}$ have nonnegative effects on the
output variable $y_{i}$. Moreover, the monotonicity is strict in,
at least, one of the three arguments $x_{i1},x_{i2}$, and $u_{i}$.
The restriction of monotonicity with respect to $\left(x_{i1},x_{i2}\right)$
is substantive: it requires that the inputs cannot negatively affect
the output variable holding everything else fixed. In contrast, the
restriction of monotonicity with respect to $u_{i}$ is largely innocuous
given the interpretation of $u_{i}$ as a (weakly) ``positive shock''.
\begin{example*}[Team Production Functions Continued]
Assumption \ref{assu:MonoF} is satisfied in the linear additive
model in equation \eqref{eq:Model_CD} provided that the model satisfies
the additional parameter restriction that $\a_{1},\a_{2}\geq0$.
\end{example*}
Next, we turn to the assumptions on the unobserved errors $u_{i}$
and $\e_{i}$ in equation \eqref{eq:Gen_model}. First, we assume
that the endogenous inputs $x_{i}$ are strictly monotone functions
of the scalar productivity shock $u_{i}$, potentially after conditioning
on a set of observed covariates $z_{i}$, that may affect the inputs
$x_{i}$.
\begin{assumption}[Strict Monotonicity of the Covariates in the Structural Shock]
 \label{assu:input_mono} There exists a vector of additional observed
covariates $z_{i}$ and two deterministic, real-valued functions $h_{1},h_{2}$,
such that
\begin{align*}
x_{i1}=h_{1}\left(u_{i},z_{i}\right),\quad & x_{i2}=h_{2}\left(u_{i},z_{i}\right),
\end{align*}
with both $h_{1}\left(u_{i},z_{i}\right)$ and $h_{2}\left(u_{i},z_{i}\right)$
strictly increasing in their first argument $u_{i}$ for every realization
of $z_{i}$.
\end{assumption}
We note that the functions $h_{1}$ and $h_{2}$ can be unknown and
nonparametric. The only requirement here is that, after conditioning
on $z_{i}$, the covariates $x_{i1}$ and $x_{i2}$ can be written
as deterministic monotone functions of the error $u_{i}$. Such a
``monotonicity-in-a-scalar-error'' assumption has been widely used
in the econometric literature on identification analysis.\footnote{See \citet{matzkin2007nonparametric} for a general survey, and see
\citet*{ackerberg2015identification} in the specific context of production
function identification, which fits into our working example \eqref{eq:Model_CD}.}
\begin{example*}[Team Production Functions Continued]
In the IO literature, $u_{i}$ is typically interpreted as a ``productivity
shock'' that enters into the choices of inputs $x_{i}$. In contrast,
$\epsilon_{i}$ captures either a measurement error or an ``ex-post
productivity shock'' that does not affect inputs, since it is not
observed to the firms when input choices are made. Assumption \ref{assu:input_mono}
requires that the input choice functions are strictly increasing in
the ``productivity shock'' $u_{i}$, conditional on any additional
observed covariates $z_{i}$ that may influence input choices, as
suggested, for example, by \citet*{olley1996dynamics}\footnote{We only consider the cross-sectional setting here, but the identification
arguments below for the imputation of partially latent inputs can
also be applied to a panel data setting with dynamic input choices,
such as in \citet*{olley1996dynamics}, where current capital input
$k_{it}$ is a function of last-period capital stock $k_{i,t-1}$
and current-period productivity shock $u_{it}$. In that setting,
lagged input $k_{i,t-1}$ will need to be part of the control variables
$z_{i,t}$.} and others. \footnote{This is a standard assumption that underlies most, if not all, existing
approaches of production function estimation in one way or another:
see, for example, \citet*{griliches1998} and \citet*{ackerberg2015identification}
for reviews of the relevant literature.} For concreteness, we take $z_{i}$ to be local wages for managers
and employees.
\end{example*}
Assumption \ref{assu:input_mono} can be further micro-founded in
a variety of settings based on efficiency or equilibrium criteria,
which we elaborate in Subsection \ref{subsec:DiscussA3}, given that
Assumption \ref{assu:input_mono} is the key assumption in this paper.
\begin{assumption}[Exogeneity of the Measurement Error]
 \label{assu:eps_error} $\mathbb{E}\left[\left.\epsilon_{i}\right|x_{i},z_{i},d_{i}\right]=0$.
\end{assumption}
Note that, under Assumption \ref{assu:input_mono}, conditioning on
$\left(x_{i},z_{i},d_{i}\right)$ is equivalent to conditioning on
$\left(u_{i},z_{i},d_{i}\right)$. In the production function estimation
literature without the partial latency problem, $\mathbb{E}\left[\left.\epsilon_{i}\right|u_{i},z_{i}\right]=0$
is a standard assumption imposed on $\e_{i}$. In our current setting,
we are requiring that $\e_{i}$ is furthermore exogenous with respect
to the partial latency indicator variable $d_{i}$.

It is worth noting that this paper is both conceptually and technically
different from previous work on missing data in linear regression
and, more generally, GMM estimation settings, such as \citet*{rubin1976inference},
\citet*{little1992regression}, \citet*{robins1994estimation}, \citet*{wooldridge2007inverse},
\citet*{graham2011efficiency}, \citet*{chaudhuri2016gmm}, \citet*{abrevaya2017gmm}
and \citet*{mcdonough2017missing}. This line of literature usually
exploits two types of assumptions to handle missing values: first,
observations with no missing data occur with positive probability,
and second, data are ``missing at random (MAR)'': the indicator
for missingness is exogenous to or independent of certain observable
covariates or constructed conditioning variables. Neither condition
is satisfied in our setting: here every observation contains ``missing
values'', and the partial latency indicator $d_{i}$ is allowed to
be correlated with other observables as well as the unobserved productivity
shock. Instead, we will be relying on monotonicity conditions to identify
and impute the latent input.

Specifically, Assumption \ref{assu:eps_error} here is simply requiring
that $\epsilon_{i}$ is a ``measurement error'' term that is exogenous
with respect to the observables and consequently the productivity
shock $u_{i}$, but does not impose any restriction on the dependence
structure between the partial latency indicator $d_{i}$ and other
structural components of the model $\left(u_{i},x_{i},z_{i}\right)$.

That said, we do require the following very mild condition on the
variable $d_{i}$.
\begin{assumption}[Nondegenerate Latency Probabilities]
 \label{assu:Prob01} $0<\P\left\{ \rest{d_{i}=1}u_{i},z_{i}\right\} <1$.
\end{assumption}
Assumption \ref{assu:Prob01} guarantees that conditioning on realizations
of $\left(u_{i},z_{i}\right)$ we will observe $x_{i1}$, and $x_{i2}$,
with strict positive probabilities. Again, this assumption is much
weaker than ``missing-at-random'' assumptions, which would usually
require that $\P\left\{ \rest{d_{i}=1}u_{i},z_{i}\right\} $ is constant
in $u_{i}$, $z_{i}$, or some other variables. In contrast, here
we do not impose any restrictions on the dependence of $\P\left\{ \rest{d_{i}=1}u_{i},z_{i}\right\} $
on $\left(u_{i},z_{i}\right)$ beyond non-degeneracy.

As discussed in the introduction, allowing for $d_{i}$ to depend
on $\left(u_{i},z_{i}\right)$, and thus $x_{i}$, is also a feature
that differentiates our monotonicity-based method from the ``moment-matching
method'' in the literature on data combination, such as \citep{angrist1992effect,ridder2007econometrics,buchinsky2022estimation},
which are based on conditional independence of $d_{i}$.

\subsection{\label{subsec:MainID}Main Result}

We are now ready to present our main identification result.
\begin{thm}
\label{thm:ID}Under Assumptions \ref{assu:OneLatent}-\ref{assu:Prob01},
for each observation $i$, the latent input, $x_{i2}$ if $d_{i}=1$
or $x_{i1}$ if $d_{i}=2$, is point identified.
\end{thm}
Given that the identification strategy underlying the Theorem \ref{thm:ID}
is the key novelty of this paper, we prove Theorem \ref{thm:ID} in
the main text below.
\begin{proof}
The starting point of our identification strategy is the reduced form
of our model with the measurement error term:
\begin{align}
y_{i} & =\ol F\left(u_{i},z_{i}\right)+\epsilon_{i}\label{eq:qi_qu_e}
\end{align}
where
\begin{equation}
\ol F\left(u_{i},z_{i}\right):=F\left(h_{1}\left(u_{i},z_{i}\right),h_{2}\left(u_{i},z_{i}\right),u_{i}\right).\label{eq:F_bar}
\end{equation}
Note that $\ol F\left(u_{i},z_{i}\right)$ is strictly increasing
in $u_{i}$ given Assumptions \ref{assu:MonoF} and \ref{assu:input_mono}.

Now, define $\gamma_{1}\left(c\right)$ as the expected output of
firm $i$ conditional on the event that $x_{i1}$ is observed ($d_{i}=1$)
to have a given value of $c_{1}$, i.e.,
\begin{equation}
\gamma_{1}\left(c_{1};z\right):=\mathbb{E}\left[\left.y_{i}\right\vert z_{i}=z,\ d_{i}=1,\ x_{i1}=c_{1}\right].\label{eq:gamma1}
\end{equation}
Note that $\gamma_{1}$ is directly identified from the data given
Assumptions \ref{assu:OneLatent} and \ref{assu:Prob01}.\footnote{Assumption \ref{assu:Prob01} ensures that the conditioning event
occurs with strictly positive probability.}

Taking a closer look at $\gamma_{1}$, we have, by equation \eqref{eq:qi_qu_e},
Assumption \ref{assu:input_mono}, and Assumption \ref{assu:eps_error},
\begin{align}
\gamma_{1}\left(c_{1};z\right) & =\mathbb{E}\left[\left.\overline{F}\left(u_{i},z_{i}\right)+\epsilon_{i}\right\vert z_{i}=z,\,d_{i}=1,\,h_{1}\left(u_{i},z_{i}\right)=c_{1}\right]\nonumber \\
 & =\overline{F}\left(h_{1}^{-1}\left(c_{1};z\right),z\right)+\E\left[\rest{\e_{i}}z_{i}=z,\,d_{i}=1,\,u_{i}=h_{1}^{-1}\left(c_{1};z\right)\right]\nonumber \\
 & =F\left(c_{1},\ h_{2}\left(h_{1}^{-1}\left(c_{1};z\right),z\right),\ h_{1}^{-1}\left(c_{1};z\right)\right).\label{eq:gamma1_detail}
\end{align}
By conditioning on $z_{i}$ and a particular \textit{observed value}
of $x_{i1}=c_{1}$, we are effectively conditioning on the \textit{unobserved}
productivity shock $u_{i}$. Aggregating across observations allows
us to average out the measurement errors and obtain a quantity that
is implicitly a function of the productivity shock $u_{i}=h_{1}^{-1}\left(c_{1};\ z_{i}\right)$.

Next, observe that $\gamma_{1}\left(c_{1};z\right)$ is strictly increasing
in $c_{1}$, since
\begin{align}
\frac{\partial}{\partial c_{1}}\gamma_{1}\left(c_{1};z\right) & =F_{1}+F_{2}\cd\frac{\partial}{\partial u}h_{2}\left(h_{1}^{-1}\left(c_{1}\right),z\right)\frac{1}{\frac{\partial}{\partial u}h_{1}\left(h_{1}^{-1}\left(c_{1}\right),z\right)}+F_{3}\cd\frac{1}{\frac{\partial}{\partial u}h_{1}\left(h_{1}^{-1}\left(c_{1}\right),z\right)}\nonumber \\
 & >0\label{eq:gamma1_deri}
\end{align}
since $\frac{\partial}{\partial u}h_{1},\frac{\partial}{\partial u}h_{2}>0$
by Assumption \ref{assu:input_mono}, and the partial derivatives
$F_{1},F_{2},F_{3}$ of $F$ are all nonnegative with, at least, one
being strictly positive by Assumption 2. This guarantees that the
inverse function $\g_{1}^{-1}\left(\cd;z\right)$ exists.\footnote{The partial derivatives $F_{1},F_{2},F_{3}$ of $F$ are evaluated
at $\left(c_{1},\ h_{2}\left(h_{1}^{-1}\left(c_{1};z\right),z_{i}\right),\ h_{1}^{-1}\left(c_{1};z\right)\right)$.}

Similarly, we can define
\[
\gamma_{2}\left(c_{2};z\right):=\mathbb{E}\left[\left.y_{i}\right|z_{i}=z,\ d_{i}=2,\ x_{i2}=c_{2}\right]
\]
which is strictly increasing in $c_{2}$ and thus invertible with
respect to its first argument.

Now, the key idea behind our identification strategy is to consider
the event that
\begin{equation}
\gamma_{1}\left(c_{1};z\right)=\gamma_{2}\left(c_{2};z\right)\label{eq:gamma_equal}
\end{equation}
for some $c_{1},c_{2}$, and $z$. By \eqref{eq:gamma1_detail}, equation
\eqref{eq:gamma_equal} holds if and only if
\begin{equation}
h_{1}^{-1}\left(c_{1};\ z\right)=h_{2}^{-1}\left(c_{2};\ z\right)=u\label{eq:h_inv_u}
\end{equation}
for some value of the productivity shock $u$, which is furthermore
equivalent to
\begin{equation}
c_{1}=h_{1}\left(u;z\right),\quad c_{2}=h_{2}\left(u;z\right)\label{eq:c_and_u}
\end{equation}
for some $u$.\footnote{To see more clearly why \eqref{eq:c_and_u} is true, consider WLOG
the possibility that $c_{1}=h_{1}\left(u_{1};z\right)$ and $c_{2}=h_{2}\left(u_{2},z\right)$
for some $u_{1}>u_{2}$. Then by \eqref{eq:gamma1_detail} we have
\begin{align*}
\gamma_{1}\left(c_{1};z\right) & =F\left(h_{1}\left(u_{1},z\right),\ h_{2}\left(u_{1},z\right),\ u_{1}\right)\\
 & >F\left(h_{1}\left(u_{2},z\right),\ h_{2}\left(u_{2},z\right),\ u_{2}\right)=\g\left(c_{2};z\right).
\end{align*}
}

Now, consider any observation $i$ with covariates $x_{i1}$ and $x_{i2}$.
Recall that $x_{i1}=h_{1}\left(u_{i};z_{i}\right)$ and $x_{i2}=h_{2}\left(u_{i};z_{i}\right)$
by Assumption \ref{assu:input_mono}. Then, by the equivalence of
\eqref{eq:c_and_u} and \eqref{eq:gamma_equal} established above,
we deduce that
\[
\gamma_{1}\left(x_{i1};z_{i}\right)=\gamma_{2}\left(x_{i2};z_{i}\right).
\]
Hence, if $x_{i1}$ is observed while $x_{i2}$ is latent, then the
latent $x_{i2}$ can be identified via a composition of $\g_{2}^{-1}$
and $\g_{1}$ as
\[
x_{i2}=\gamma_{2}^{-1}\left(\gamma_{1}\left(x_{i1};z_{i}\right);z_{i}\right),
\]
where on the right-hand side $x_{i1}$ is observed and $\gamma_{1},\gamma_{2}$
are nonparametrically identified functions. Similarly, $x_{i1}$ can
be identified when $x_{i2}$ is observed.

In summary, we can identify the partially latent covariate by
\begin{align}
x_{i2} & =\gamma_{2}^{-1}\left(\gamma_{1}\left(x_{i1};z_{i}\right);z_{i}\right),\quad\text{for }d_{i}=1,\nonumber \\
x_{i1} & =\gamma_{1}^{-1}\left(\gamma_{2}\left(x_{i2};z_{i}\right);z_{i}\right),\quad\text{for }d_{i}=2.\label{eq:ID_latent}
\end{align}
\end{proof}
It should be pointed out that \eqref{eq:ID_latent} is an explicit
representation of the ``functional dependence" between the two input
variables as in \citet{ackerberg2015identification}: $x_{i1}$ is
a deterministic function of $x_{i2}$, and vice versa, conditional
on instruments $z_{i}$. While functional dependence was a concern
in the context of \citet{olley1996dynamics}, \citet{levinsohn2003estimating}
and \citet{ackerberg2015identification}, here we are exactly leveraging
the functional dependence between input variables to solve the partially
latency problem.
\begin{rem}[More Than Two Inputs]
\label{rem:more_inputs} We have thus far focused on the case with
two inputs. It is straightforward to see that our model, assumptions,
and the main identification result can be easily generalized to the
case with inputs of an arbitrary finite dimension $D$. This result
is summarized by the following Corollary.
\end{rem}
\begin{cor}
\label{cor:more_inputs} Consider the model $y_{i}:=F\left(x_{i1},...,x_{iD},u_{i}\right)+\e_{i}$
along with Assumptions \ref{assu:MonoF} and \ref{assu:eps_error}
unchanged, and the following modifications of other assumptions:
\begin{itemize}
\item[(i)] Assumption \ref{assu:OneLatent}: for each $i$ at least one out
of $D$ inputs is observed.
\item[(ii)] Assumption \ref{assu:input_mono}: all $D$ inputs are strictly increasing
in $u_{i}$ given $z_{i}$.
\item[(iii)] Assumption \ref{assu:Prob01}: all $D$ inputs are observed with
strictly positive probabilities.
\end{itemize}
Then the latent inputs are identified.
\end{cor}
\begin{rem}
If Condition (i) in Corollary \ref{cor:more_inputs} is strengthened
so that \emph{more than one} inputs are simultaneously observed in
a given observation (with positive probability), then we would also
obtain over-identification, and the input-monotonicity restriction
in Assumption \ref{assu:input_mono} becomes empirically refutable.
Alternatively, with two or more inputs simultaneously observed, we
would be able to accommodate higher dimensions of unobserved shocks,
provided that the dimension of the unobserved shock $u_{i}$ is strictly
smaller than the dimension of the covariates $D$. Since such an extension
would be more involved and move farther away from the applications
we consider in this paper, we leave it as a direction for future research.
\end{rem}

\subsection{\label{subsec:DiscussA3}Discussion of Assumption \ref{assu:input_mono}}

The monotonicity of input choices in the unobserved productivity shock
(Assumption \ref{assu:input_mono}) can be further micro-founded in
a variety of settings based on efficiency or equilibrium criteria.

On a general level, one may use the theory of monotone comparative
statics to obtain more primitive conditions for input monotonicity,
which typically involve various forms of supermodularity (or increasing-difference)
conditions: see \citet*{topkis1998supermodularity} and \citet*{vives2000oligopoly}
for general treatments on this topic. Essentially, in settings where
input choices are made by a single decision maker, we need the objective
function to be supermodular in input variables $x$ and the productivity
shock $u$. In settings where the input choices are generated as equilibria
of a strategic game between two decision makers, strategic complementarity
is typically required (so that the game is supermodular) to establish
monotonicity. For games with strategic substitutes, we further need
a condition to ensure that the extent of strategic substitutes is
not overwhelming: see, for example, \citet*{roy2010monotone}.

We now provide some concrete examples that illustrate how Assumption
\ref{assu:input_mono} can be satisfied: as the optimal choice of
a single decision maker, as the Nash equilibrium of a game of strategic
complements, as the Nash equilibrium of a game of strategic substitutes.

\subsubsection*{Optimal Choice of a Single Decision Maker}

Suppose that firms optimally choose inputs to maximize profits under
the Cobb-Douglas production function with a constant output price.
Formally, each firm $i$ solves the problem
\[
\max_{X_{i1},X_{i2}}e^{\alpha_{0}+u_{i}}X_{i1}^{\alpha_{1}}X_{i2}^{\alpha_{2}}-Z_{i1}X_{i1}-Z_{i2}X_{i2},
\]
where $X_{i1},X_{i2}$ are inputs in its original scale (with $x_{i1},x_{i2}$
denoting the logarithm of $X_{i1},X_{i2}$) and $Z_{i1},Z_{i2}$ are
input prices in its original scale (with $z_{i1},z_{i2}$ denoting
the logarithm of $Z_{i1},Z_{i2}$). Then the input choice functions
$h_{1}$ and $h_{2}$ are characterized by the relevant first-order
conditions and have simple closed-form solutions that are linear,
increasing in $u_{i}$, and decreasing in $z_{i}$.\footnote{We note that the problem of partially latent inputs is less relevant
in that case since the ``reduced-form'' regression of the observed
inputs on the exogenous wages $w_{i}$ will indirectly recover the
production function parameters $\alpha$. This corresponds to the
``duality approach'' to production function estimation as discussed
in detail in \citet*{griliches1998}. However, an attractive feature
of our approach is also that we can test whether inputs are optimally
chosen. If we reject the null hypothesis that inputs are optimal,
our estimator is still feasible while duality estimators are not.} In particular, we have:
\[
h_{1}\left(u,z\right)=\frac{\alpha_{0}+\left(1-\alpha_{2}\right)\log\alpha_{1}+\alpha_{2}\log\alpha_{2}-\left(1-\a_{2}\right)z_{1}-\left(1-\alpha_{2}\right)z_{2}+u}{1-\a_{1}-\a_{2}},
\]
satisfying Assumption \ref{assu:input_mono}. See Appendix \ref{subsec:CD_Optimal}
for more details of the derivation.

\subsubsection*{Nash Equilibrium in a Game of Strategic Complements}

Assumption \ref{assu:input_mono} is also satisfied if the inputs
are Nash equilibrium choices of two partners, each of whom solves
the following optimization problem: given $u,z$ and the other partner's
choice $X_{2}$, partner 1 solves
\[
\max_{X_{1}}\pi_{1}\left(X,u;Z\right):=\l_{1}\left(F\left(X,u\right)-Z_{1}X_{1}-Z_{2}X_{2}\right)+Z_{1}X_{1}-\frac{1}{2}c_{1}X_{1}^{2},
\]
where $F\left(X,u\right):=e^{u+\a_{0}}X_{1}^{\a_{1}}X_{2}^{\a_{2}}$
is the Cobb-Douglas production function. The term $F\left(X,u\right)-Z_{1}X_{1}-Z_{2}X_{2}$
is the profit of the firm (as in the single decision maker's problem
described above), and $\l_{1}\in\left(0,1\right)$ is a positive share
of firm profit distributed to partner 1 as ``dividends''. Moreover,
in addition to the ``dividends'', partner 1 receives her wage income
$Z_{1}X_{1}$. Finally, $\frac{1}{2}c_{1}X_{1}^{2}$ captures partner
1's quadratic private cost of input $X_{1}$ with $c_{1}>0$.

Similar, partner 2 solves
\[
\max_{X_{2}}\pi_{2}\left(X,u;Z\right):=\l_{2}\left(F\left(X,u\right)-Z_{1}X_{1}-Z_{2}X_{2}\right)+Z_{2}X_{2}-\frac{1}{2}c_{2}X_{2}^{2},
\]
with $\l_{2}\in\left(0,1\right)$ and $c_{2}>0$.

This is a (supermodular) game of strategic complementarity since
\begin{equation}
\Dif_{X_{1}X_{2}}\pi_{j}=\l_{j}\Dif_{X_{1}X_{2}}F=\l_{j}\a_{1}\a_{2}e^{\a_{0}+u}X_{1}^{\a_{1}-1}X_{2}^{\a_{2}-1}>0.\label{eq:eg_strat_comp}
\end{equation}
Furthermore, the payoff functions feature increasing differences between
the productivity shock $u$ and the input $x_{j}$:
\begin{equation}
\Dif_{uX_{j}}\pi_{j}=\l_{j}\Dif_{uX_{j}}F=\a_{j}e^{\a_{0}+u}X_{j}^{\a_{j}-1}X_{k}^{\a_{k}}>0.\label{eq:eg_inc_diff}
\end{equation}

With these conditions and the theory on monotone comparative statics,
say, \citet*{milgrom1990rationalizability}, we can show that the
unique Nash equilibrium $X^{*}\left(u,Z\right)$ of this game is strictly
increasing in $u$, thus satisfying Assumption \ref{assu:input_mono}.
See Appendix \ref{subsec:CD_StratComp} for the detailed proof.

\subsubsection*{Nash Equilibrium in a Game of Strategic Substitutes}

Consider an example with a married couple and two parental households,
$j=1,2$, whose wealth levels are respectively $Z_{1}$ and $Z_{2}$,
which is based on \citet*{Bergstrom-Blume-Varian-86}. Parents are
altruistic toward their married offspring but not toward that offspring's
spouse. Parental household $j$ has utility
\[
v_{j}\left(X,Z\right)=\log\left(Z_{j}-X_{j}\right)+u\log\left(X_{1}+X_{2}\right)
\]
where $X_{j}$ is the married couple's gift from parental household
$j$ and $u$ is the probability that both parental households think
the children's marriage will endure. This leads to a noncooperative
game between the two parental households since the incentive for either
household to gift the offspring couple diminishes as the other parental
household gives more. Formally, parental household $j$'s marginal
return from $X_{j}$
\[
\frac{\p}{\p X_{j}}v_{j}\left(X,Z\right)=u\cd\frac{1}{X_{j}+X_{k}}-\frac{1}{Z_{j}-X_{j}}
\]
is decreasing in the other parental household's return $X_{k}$, and
thus the best response of household $j$
\[
BR_{j}\left(X_{k}\right)=\frac{u}{1+u}Z_{j}-\frac{1}{1+u}X_{k}
\]
is also decreasing in $X_{k}$. Hence, this is a game of strategic
substitutes.

There is a unique Nash equilibrium of this game between the two parental
households, given by
\[
X_{1}^{\ast}=\frac{(1+u)\;Z_{1}-Z_{2}}{2+u},\quad X_{2}^{\ast}=\frac{(1+u)\;Z_{2}-Z_{1}}{2+u}
\]
for any $u$ and wealth levels $Z_{1},Z_{2}$, provided that the two
households that are not ``too'' different in wealth so that interior
solutions in $X_{1}^{*},X_{2}^{*}$ obtain.\footnote{Formally, for the interiority we also need to require that $u$ is
strictly positive and bounded away from 0 in this stylized example.
One can perturb the example in various ways to ensure interiority
without such a restriction, but at the cost of additional complications.} Both $X_{1}^{*}$ and $X_{2}^{*}$ are strictly increasing in the
shock $u$, and hence the outcome is strictly increasing in $u$.

\section{\label{sec:Estimation} Estimation of the Outcome Function}

Once the partially latent covariates $x_{i1},x_{i2}$ are identified
and imputed, researchers may use them to identify and estimate functions
or parameters defined based on $\left(x_{i1},x_{i2}\right)$. Note
that researchers may use $\left(x_{i1},x_{i2}\right)$ for completely
different purpose from the identification and estimation of the outcome
function $F$. Hence, our proposed method in Section \ref{sec:Identification}
can be thought as a monotonicity-based method for data imputation
or data combination.

That said, how to identify and estimate the outcome function $F$
is a natural question to ask, given that our method for the identification
of partially latent inputs is built upon assumptions on $F$. Hence
we focus on the identification and estimation of $F$ in this section.

With the latent inputs identified in Theorem \ref{thm:ID}, we are
back to equation \eqref{eq:Gen_model}
\[
y_{i}=F\left(x_{i1},x_{i2},u_{i}\right)+\e_{i},
\]
but now we can effectively regard both $x_{i1}$ and $x_{i2}$ as
being known, at least for identification purposes. Researchers may
proceed to identify the production function $F$ under appropriate
application-specific assumptions as in a ``standard'' setting without
the partial latency problem. Hence, the identification of $F$ or
other objects of interest is largely ``separable'' from the partial
latency problem, which is the key problem we are solving in this paper.

That said, we note that the \textit{estimation} of the latent inputs
will affect the \textit{estimation} of (the parameters of) $F$ based
on ``plugged-in'' latent input estimates. This section provides
a discussion on how to identify and estimate $F$, and analyzes the
impact of the ``first-stage'' estimation of latent inputs on the
final estimator of $F$.

While we cannot cover all relevant specifications of $F$, in this
section we will provide both identification and estimation results
for the linear case, which is arguably the workhorse model, or at
least a natural benchmark, in various empirical applications. We also
discuss how our method can be applied under more general settings.

\subsection{\label{subsec:Asymptotics}The Linear Model}

In this subsection we focus on the linear parametric specification
of $F$ as in \eqref{eq:Model_CD}:
\[
y_{i}=\alpha_{0}+\alpha_{1}x_{i1}+\alpha_{2}x_{i2}+u_{i}+\e_{i},
\]
where our goal is to identify and estimate the unknown parameters
$\alpha:=\left(\alpha_{0},\alpha_{1},\alpha_{2}\right)$.

\subsubsection{Identification}

In the presence of the endogeneity problem between $x_{i}:=\left(x_{i1},x_{i2}\right)$
and $u_{i}$, we will need instrumental variables for the identification
of $\a$. For illustrational simplicity, we impose the following standard
IV assumption.
\begin{assumption}[Instrumental Variables]
 \label{assu:IVz} Write $z_{i}:=\left(z_{i1},z_{i2}\right)$, $\overline{z}_{i}:=\left(1,z_{i1},z_{i2}\right)^{^{\prime}}$
and $\ol x_{i}=\left(1,x_{i1},x_{i2}\right)^{^{\prime}}$. Assume
\begin{itemize}
\item[(i)] Relevance: $\Sigma_{zx}:=\mathbb{E}\left[\overline{z}_{i}\overline{x}_{i}^{^{\prime}}\right]$
has full rank.
\item[(ii)] Exogeneity: $\mathbb{E}\left[\left.u_{i}\right|z_{i}\right]=0$.
\end{itemize}
\end{assumption}
\begin{rem}
Here we are using the same ``$z_{i}$'', i.e., the observable determinants
of input choices in Assumption \ref{assu:input_mono}, as the instrumental
variables for identification and estimation of the outcome equation,
because such ``$z_{i}$'' are naturally relevant (as it enters into
the input choice function $h$) and excluded (as it does not enter
into the outcome equation $F$) given our model specification.

However, it should be pointed out that we can also include additional
instrumental variables beyond the ``$z_{i}$'' in the input choice
functions. These additional IVs from outside our model can help with
the identification and estimation of the outcome function as well,
as long as the relevance and exogeneity conditions in Assumption \ref{assu:IVz}
are satisfied.

This point also illustrates the ``separability'' between the problem
of identifying partially latent variables in Section \ref{sec:Identification}
and the problem of identifying the outcome function in Section \ref{sec:Estimation}.
Correspondingly, we view our framework and results in Section \ref{sec:Identification}
as the main contribution of this paper, while the current Section
\ref{sec:Estimation} is just one example of how our main results
in \ref{sec:Identification} can be used.
\end{rem}
\begin{cor}[Identification of Linear Parameters]
Under Assumptions \eqref{assu:OneLatent}-\eqref{assu:IVz}, $\a$
is point identified.
\end{cor}
\begin{example*}[Team Production Function Continued]
In the context of our working example, here we are essentially following
a strategy discussed in \citet*{griliches1998} and assume that we
have access to some instrumental variables (such as local wages) that
affect input choices.
\end{example*}

\subsubsection{\label{subsec:EstSimple}Estimation Procedure}

We now turn to the more interesting problem of estimation, propose
semiparametric estimators for $\a$, and characterize their asymptotic
distributions.

We first describe our proposed estimator. Since the identification
of latent inputs via equation \eqref{eq:ID_latent} is constructive,
it suggests a natural estimation procedure:

\medskip{}

\textbf{Step 1} (Nonparametric Regression): obtain an estimator $\hat{\gamma}_{1}$
of $\gamma_{1}$ by nonparametrically regressing $y_{i}$ on $x_{i1}$
and $z_{i}$, among firms with $d_{i}=1$, i.e., those with $x_{i1}$
observed. Similarly, obtain an estimator $\hat{\gamma}_{2}$ of $\gamma_{2}$.

\textbf{Step 2} (Imputation): impute latent inputs by plugging the
nonparametric estimators $\hat{\gamma}_{1},\hat{\gamma}_{2}$ into
equation \eqref{eq:ID_latent}, i.e.,
\begin{align*}
\hat{x}_{i2} & =\hat{\gamma}_{2}^{-1}\left(\hat{\gamma}_{1}\left(x_{i1};z_{i}\right);z_{i}\right),\quad\text{for }d_{i}=1,\\
\hat{x}_{i1} & =\hat{\gamma}_{1}^{-1}\left(\hat{\gamma}_{2}\left(x_{i2};z_{i}\right);z_{i}\right),\quad\text{for }d_{i}=2.
\end{align*}

\textbf{Step 3} (IV Regression): run either of the following two IV
regressions:
\begin{itemize}
\item[(3a)] Estimate equation \eqref{eq:Model_CD} with $z_{i}$ as IVs for $x_{i}$,
i.e.,
\[
\hat{\alpha}:=\left(\frac{1}{n}\sum_{i=1}^{n}\overline{z}{}_{i}\tilde{x}_{i}\right)^{-1}\left(\frac{1}{n}\sum_{i=1}^{n}\overline{z}{}_{i}y_{i}\right)
\]
where $\overline{z}_{i}:=\left(1,z_{i1},z_{i2}\right)^{^{\prime}}$
and
\[
\tilde{x}_{i}:=\begin{cases}
\left(1,x_{i1},\hat{x}_{i2}\right)^{^{\prime}}, & \text{ for }d_{i}=1,\\
\left(1,\hat{x}_{i1},x_{i2}\right)^{^{\prime}}, & \text{ for }d_{i}=2.
\end{cases}
\]
\item[(3b)] Estimate the following equation
\begin{equation}
\overline{y}_{i}=\alpha_{0}+\alpha_{1}x_{i1}+\alpha_{2}x_{i2}+u_{i},\label{eq:CD_exp}
\end{equation}
with the outcome variable
\[
\overline{y}_{i}:=\overline{F}\left(u_{i},z_{i}\right)=\gamma_{1}\left(x_{i1},z_{i}\right)=\gamma_{2}\left(x_{i2},z_{i}\right),
\]
replaced by its plug-in estimator
\[
\tilde{y}_{i}:=\begin{cases}
\hat{\gamma}_{1}\left(x_{i1},z_{i}\right), & \text{ for }d_{i}=1,\\
\hat{\gamma}_{2}\left(x_{i2},z_{i}\right), & \text{ for }d_{i}=2,
\end{cases}
\]
Again using $z_{i}$ as IVs, estimate $\a$ by
\[
\hat{\alpha}^{*}:=\left(\frac{1}{n}\sum_{i=1}^{n}\overline{z}{}_{i}\tilde{x}_{i}\right)^{-1}\left(\frac{1}{n}\sum_{i=1}^{n}\overline{z}_{i}\tilde{y}_{i}\right).
\]
\end{itemize}

\subsubsection{\label{subsec:AsymTheory}Asymptotic Theory}

We now establish the consistency and the asymptotic normality of $\hat{\alpha}$
and $\hat{\a}^{*}$ under the following regularity assumptions.
\begin{assumption}[Finite Error Variances]
\label{assu:ErrVar} $\mathbb{E}\left[\left.u_{i}^{2}\right|z_{i}\right]<\infty$
and $\mathbb{E}\left[\left.\epsilon_{i}^{2}\right|x_{i},z_{i},d_{i}\right]<\infty$.
\end{assumption}
\begin{assumption}[Strong Monotonicity]
 \label{assu:StrongMono} The first derivative of $\gamma_{k}\left(\cdot,z\right)$
is uniformly bounded away from zero, i.e., for any $c,z$,
\[
\frac{\partial}{\partial c}\gamma_{k}\left(c;z\right)>\underline{c}>0.
\]
\end{assumption}
In view of equation \eqref{eq:gamma1_deri}, Assumption \ref{assu:StrongMono}
is satisfied if either $\alpha_{1},\alpha_{2}>0$ or $\frac{\partial}{\partial u}h_{1},\frac{\partial}{\partial u}h_{1}$
are uniformly bounded above by a finite constant. Assumption \ref{assu:StrongMono}
is needed to ensure that $\hat{\gamma}_{k}^{-1}\left(\cdot,z\right)$
is a good estimator of $\gamma_{k}^{-1}\left(\cdot,z\right)$ provided
that the first-stage nonparametric estimator $\hat{\gamma}_{k}$ is
consistent for $\gamma_{k}$.
\begin{assumption}[First-Stage Estimation]
 \label{assu:SP_RegCon}~
\begin{itemize}
\item[(i)] Donsker property: $\gamma_{1},\gamma_{2}\in\Gamma$, which is a Donsker
class of functions with uniformly bounded first and second derivatives,
and $\hat{\gamma}_{1},\hat{\gamma}_{2}\in\Gamma$ with probability
approaching 1.
\item[(ii)] First-stage convergence: $\left\Vert \hat{\gamma}_{k}-\gamma_{k}\right\Vert _{\infty}=o_{p}\left(N^{-\frac{1}{4}}\right)$
for $k=1,2$.
\end{itemize}
\end{assumption}
Assumption \ref{assu:SP_RegCon}(i) is guaranteed if $\gamma_{1},\gamma_{2}$
satisfy certain smoothness condition, e.g. $\gamma_{k}$ possesses
uniformly bounded derivatives up to a sufficiently high order. Assumption
\ref{assu:SP_RegCon}(ii) requires that the first-stage estimator
converges at a rate faster than $N^{-1/4}$, which is satisfied under
various types of nonparametric estimators under certain regularity
conditions. This is required so that the final estimator of the production
function parameters $\alpha$ can converge at the standard parametric
($\sqrt{N}$) rate despite the slower first-step nonparametric estimation
of $\gamma_{1},\gamma_{2}$.

Finally, we state another technical assumption that captures how the
first-stage nonparametric estimation of $\gamma_{1},\gamma_{2}$ influences
the final semiparametric estimators $\hat{\alpha}$ and $\hat{\a}^{*}$
through the functional derivatives of the residual function with respect
to $\gamma_{1},\gamma_{2}$. Assumption \ref{assu:AsymLin} below,
based on \citet*{newey1994asymptotic}, provides an explicit formula
for the asymptotic variances of $\hat{\alpha}$ and $\hat{\a}^{*}$
that does not depend on the particular forms of first-stage nonparametric
estimators.

Formally, write $w_{i}:=\left(y_{i},x_{i},z_{i},d_{i}\right)$, $\gamma:=\left(\gamma_{1},\gamma_{2}\right)$,
and suppress the conditioning variables $z_{i}$ in $\g$ for notational
simplicity. Define the residual functions
\begin{align*}
g\left(w_{i},\tilde{\alpha},\tilde{\gamma}\right) & :=\begin{cases}
\overline{z}_{i}\left(y_{i}-\tilde{\alpha}_{0}-\tilde{\alpha}_{1}x_{i1}-\tilde{\alpha}_{2}\tilde{\gamma}_{2}^{-1}\left(\tilde{\gamma}_{1}\left(x_{i1}\right)\right)\right) & \text{for }d_{i}=1,\\
\overline{z}_{i}\left(y_{i}-\tilde{\alpha}_{0}-\tilde{\alpha}_{2}x_{i2}-\tilde{\alpha}_{1}\tilde{\gamma}_{1}^{-1}\left(\tilde{\gamma}_{2}\left(x_{i2}\right)\right)\right) & \text{for }d_{i}=2.
\end{cases}\\
g^{\ast}\left(w_{i},\tilde{\alpha},\tilde{\gamma}\right) & :=\begin{cases}
\overline{z}_{i}\left(\tilde{\gamma}_{1}\left(x_{i1}\right)-\tilde{\alpha}_{0}-\tilde{\alpha}_{1}x_{i1}-\tilde{\alpha}_{2}\tilde{\gamma}_{2}^{-1}\left(\tilde{\gamma}_{1}\left(x_{i1}\right)\right)\right) & \text{for }d_{i}=1,\\
\overline{z}_{i}\left(\tilde{\gamma}_{2}\left(x_{i2}\right)-\tilde{\alpha}_{0}-\tilde{\alpha}_{2}x_{i2}-\tilde{\alpha}_{1}\tilde{\gamma}_{1}^{-1}\left(\tilde{\gamma}_{2}\left(x_{i2}\right)\right)\right) & \text{for }d_{i}=2
\end{cases}
\end{align*}
for generic $\tilde{\alpha},\tilde{\gamma}$, and
\[
g\left(w_{i},\tilde{\gamma}\right):=g\left(w_{i},\alpha,\tilde{\gamma}\right),\quad g^{\ast}\left(w_{i},\tilde{\gamma}\right):=g^{\ast}\left(w_{i},\alpha,\tilde{\gamma}\right),
\]
at the true $\alpha$. Define the pathwise functional derivative of
$g$ at $\gamma$ along direction $\tau$ by
\[
G\left(w_{i},\tau\right):=\lim_{t\rightarrow0}\frac{1}{t}\left[g\left(w_{i},\gamma+t\tau\right)-g\left(w_{i},\gamma\right)\right],
\]
and similarly define $G^{\ast}\left(z_{i},\tau\right)$ for $g^{\ast}$.
Then, following \citet*{newey1994asymptotic}, the influence function
can be derived analytically\footnote{See the proof of Theorem \ref{thm:AsymDist} for details on the calculation.}
based on $G$ and takes the form of $\varphi\left(w_{i}\right)\overline{z}_{i}\epsilon_{i}$
with
\begin{align*}
\varphi\left(w_{i}\right) & :=-\left(\lambda_{1}\frac{\alpha_{2}}{\gamma_{2}^{^{\prime}}}-\lambda_{2}\frac{\alpha_{1}}{\gamma_{1}^{^{\prime}}}\right)\left(\mathbf{\mathbbm1}\left\{ d_{i}=1\right\} -\mathbf{\mathbbm1}\left\{ d_{i}=2\right\} \right).\\
\varphi^{\ast}\left(w_{i}\right) & :=\left[\lambda_{1}\left(1-\frac{\alpha_{2}}{\gamma_{2}^{^{\prime}}}\right)+\lambda_{2}\frac{\alpha_{1}}{\gamma_{1}^{^{\prime}}}\right]\mathbf{\mathbbm1}\left\{ d_{i}=1\right\} +\left[\lambda_{1}\frac{\alpha_{2}}{\gamma_{2}^{^{\prime}}}+\lambda_{2}\left(1-\frac{\alpha_{1}}{\gamma_{1}^{^{\prime}}}\right)\right]\mathbf{\mathbbm1}\left\{ d_{i}=2\right\}
\end{align*}
where $\gamma_{k}^{^{\prime}}$ denotes $\frac{\partial}{\partial h_{k}}\gamma_{k}\left(x_{ik};z_{i}\right)$,
$\lambda_{1}$ stands for
\[
\lambda_{1}\left(x_{i},z_{i}\right):=\mathbb{E}\left[\left.\mathbf{\mathbbm1}\left\{ d_{i}=1\right\} \right|x_{i},z_{i}\right]
\]
i.e., the conditional probability of observing $x_{i1}$, and $\lambda_{2}:=1-\lambda_{1}$.

The influence function essentially characterizes how the first-stage
estimation influences the asymptotic variance of the final estimator.
Formally, we present the following assumption, commonly known as an
asymptotic linearity condition, which basically requires that the
expected error induced by the first-stage estimation is asymptotically
equivalent to the sample averages of $\varphi\left(w_{i}\right)\overline{z}_{i}\epsilon_{i}$
and $\varphi^{*}\left(w_{i}\right)\overline{z}_{i}\epsilon_{i}$.
In particular, the formula for $\varphi$ and $\varphi^{*}$ given
above will be the same regardless of the specific forms of first-step
estimators used, provided that some suitable regularity conditions
are satisfied.
\begin{assumption}[Asymptotic Linearity]
\label{assu:AsymLin} ~
\begin{itemize}
\item[(i)] Suppose
\[
\int G\left(w,\hat{\gamma}-\gamma\right)d\mathbb{P}\left(w\right)=\frac{1}{N}\sum_{i=1}^{N}\varphi\left(w_{i}\right)\overline{z}_{i}\epsilon_{i}+o_{p}\left(N^{-\frac{1}{2}}\right).
\]
\item[(ii)] Suppose
\[
\int G^{*}\left(w,\hat{\gamma}-\gamma\right)d\mathbb{P}\left(w\right)=\frac{1}{N}\sum_{i=1}^{N}\varphi^{*}\left(w_{i}\right)\overline{z}_{i}\epsilon_{i}+o_{p}\left(N^{-\frac{1}{2}}\right).
\]
\end{itemize}
\end{assumption}
We emphasize that Assumptions \ref{assu:SP_RegCon} and \ref{assu:AsymLin}
are standard assumptions widely imposed in the semiparametric estimation
literature, which can be satisfied by many kernel or sieve first-stage
estimators under a variety of conditions. See \citet*{newey1994asymptotic},
\citet*{newey1994large} and \citet*{chen2003estimation} for references.
In Assumption \ref{assu:KernelConds} below, we also provide an example
of lower-level conditions that replace Assumptions \ref{assu:SP_RegCon}
and \ref{assu:AsymLin} when we use the Nadaraya-Watson kernel estimator
in the first-stage nonparametric regression.

The next theorem establishes the asymptotic normality of $\hat{\alpha}$.
\begin{thm}[Asymptotic Normality]
\label{thm:AsymDist}Suppose Assumptions \ref{assu:OneLatent}-\ref{assu:SP_RegCon}
hold.
\begin{itemize}
\item[(i)] With Assumption \ref{assu:AsymLin}(i),
\[
\sqrt{N}\left(\hat{\alpha}-\alpha\right)\overset{d}{\longrightarrow}\mathcal{N}\left(\mathbf{0},\Sigma\right),
\]
where $\Sigma:=\Sigma_{zx}^{-1}\Omega\Sigma_{xz}^{-1}$ and
\[
\Omega:=\mathbb{E}\left[\overline{z}_{i}\overline{z}_{i}^{^{\prime}}\left(u_{i}^{2}+\left[1+\varphi\left(w_{i}\right)\right]^{2}\epsilon_{i}^{2}\right)\right].
\]
\item[{\Large{}{}(ii)}] With Assumption \ref{assu:AsymLin}(ii),
\[
\sqrt{N}\left(\hat{\alpha}^{*}-\alpha^{*}\right)\overset{d}{\longrightarrow}\mathcal{N}\left(\mathbf{0},\Sigma^{*}\right),
\]
where $\Sigma^{*}:=\Sigma_{zx}^{-1}\Omega^{*}\Sigma_{xz}^{-1}$ and
\[
\Omega^{*}:=\mathbb{E}\left[\overline{z}_{i}\overline{z}_{i}^{^{\prime}}\left(u_{i}^{2}+\varphi^{*}\left(z_{i}\right)^{2}\epsilon_{i}^{2}\right)\right].
\]
\end{itemize}
\end{thm}
We note that, if the latent inputs were observed and the first-step
nonparametric regression were not required, the asymptotic variance
of standard IV estimator of $\alpha$ would be given by $\Sigma_{zx}^{-1}\text{Var}\left(\overline{z}_{i}\left(u_{i}+\epsilon_{i}\right)\right)\Sigma_{xz}^{-1}$.
Hence, the presence of the additional term $\varphi\left(z_{i}\right)$
in $\Omega$ captures the effect of the first-step nonparametric regression
on the asymptotic variance of $\hat{\alpha}$.

To obtain consistent variance estimators, define
\begin{align*}
\hat{\Omega} & :=\frac{1}{N}\sum_{i=1}^{N}\overline{z}_{i}\overline{z}_{i}^{^{\prime}}\left[y_{i}-\tilde{x}_{i}^{^{\prime}}\hat{\alpha}+\hat{\varphi}\left(w_{i}\right)\left(y_{i}-\tilde{y}_{i}\right)\right]^{2}
\end{align*}
where
\[
\tilde{y}_{i}:=\begin{cases}
\hat{\gamma}_{1}\left(x_{i1},z_{i}\right), & \text{ for }d_{i}=1,\\
\hat{\gamma}_{2}\left(x_{i2},z_{i}\right), & \text{ for }d_{i}=2,
\end{cases}
\]
and with
\begin{align*}
\hat{\varphi}\left(w_{i}\right) & :=-\left(\hat{\lambda}_{1}\frac{\hat{\alpha}_{2}}{\hat{\gamma}_{2}^{^{\prime}}}-\hat{\lambda}_{2}\frac{\hat{\alpha}_{1}}{\hat{\gamma}_{1}^{^{\prime}}}\right)\left(\mathbf{\mathbbm1}\left\{ d_{i}=1\right\} -\mathbf{\mathbbm1}\left\{ d_{i}=2\right\} \right)
\end{align*}
where $\hat{\lambda}_{1}$ is any consistent nonparametric estimator
of $\l_{1}$. Then the variance estimators can be obtained as
\[
\hat{\Sigma}:=S_{x\tilde{z}}^{-1}\hat{\Omega}S_{\tilde{z}x}^{-1}
\]
with $S_{z\tilde{x}}:=\frac{1}{N}\sum_{i=1}^{N}\overline{z}_{i}\tilde{x}_{i}^{^{\prime}}$.

Similarly, $\hat{\O}^{*}$ and $\hat{\Sigma}^{*}$ can be constructed
accordingly.
\begin{cor}
\label{prop:VarEst}In addition to Assumptions \ref{assu:OneLatent}-\ref{assu:AsymLin},
suppose that $\hat{\lambda}_{1}$ is any consistent nonparametric
estimator of $\l_{1}$. Then $\hat{\Sigma}\overset{p}{\longrightarrow}\Sigma$
and $\hat{\Sigma}^{*}\overset{p}{\longrightarrow}\Sigma^{*}.$
\end{cor}
If furthermore $\lambda_{1}\left(x_{i},z_{i}\right)\equiv\lambda_{1}\in\left(0,1\right)$
is assumed, then we may use the sample proportion $\hat{\l}_{1}:=\frac{1}{N}\sum_{i}\left\{ d_{i}=1\right\} $.

We also note that, when the first step takes the form of sieve lease
squares, the simple procedure in \citet*{ackerberg2012practical},
in which we may ``pretend'' that the first stage is a parametric
model, may be applied to obtain estimates of the asymptotic variance
matrix.

~

Next, we compare the asymptotic variances of $\hat{\a}^{*}$ and $\hat{\a}$,
and show that $\hat{\a}^{*}$ is in fact asymptotically more efficient.
\begin{thm}[$\hat{\a}^{*}$ is Asymptotically More Efficient than $\hat{\a}$]
\label{thm:StarBetter}$\O-\O^{*}$ is positive definite, which implies
that $\hat{\alpha}^{*}$ is asymptotically more efficient than $\hat{\alpha}$.
\end{thm}
The proof is in Appendix \ref{subsec:pf_StarBetter}. Here we discuss
the intuition of Theorem \ref{thm:StarBetter}. The error term for
the IV regression with the raw outcome $y_{i}$ as the left-hand-side
variable is $u_{i}+\epsilon_{i}$, which has a larger variance than
the corresponding error term $u_{i}$, if the conditionally expected
outcome $\overline{y}_{i}$ is used instead. Even though we do not
observe $\overline{y}_{i}$ and must use an estimator $\tilde{y}_{i}=\hat{\gamma}_{1}\left(x_{i1}\right)$
or $\tilde{y}_{i}=\hat{\gamma}_{2}\left(x_{i2}\right)$, the impact
of the first-stage estimation error (which can be loosely thought
as an average of $\epsilon_{i}$ across $i$) is smaller than the
impact of $\epsilon_{i}$ itself.

To see this more clearly, first consider the multiplier ``$1+\varphi\left(w_{i}\right)$''
in (i): the ``1'' comes from the one ``raw'' share of error $\epsilon_{i}$
embedded in each $y_{i}$ that we use as the outcome variable, while
``$\varphi\left(w_{i}\right)$'' essentially captures the share
of influence of the first-step estimation error $\hat{\gamma}-\gamma$
due to $\epsilon_{i}$. Together, we have
\[
1+\varphi=\left(1-\lambda_{1}\frac{\alpha_{2}}{\gamma_{2}^{^{\prime}}}+\lambda_{2}\frac{\alpha_{1}}{\gamma_{1}^{^{\prime}}}\right)\mathbf{\mathbbm1}\left\{ d_{i}=1\right\} +\left(\lambda_{1}\frac{\alpha_{2}}{\gamma_{2}^{^{\prime}}}+1-\lambda_{2}\frac{\alpha_{1}}{\gamma_{1}^{^{\prime}}}\right)\mathbf{\mathbbm1}\left\{ d_{i}=2\right\} ,
\]
while the corresponding multiplier $\varphi^{*}$ on $\epsilon_{i}$
in (ii) is essentially the same except that ``$1-\lambda_{1}\frac{\alpha_{2}}{\gamma_{2}^{^{\prime}}}$''
becomes ``$\lambda_{1}-\lambda_{1}\frac{\alpha_{2}}{\gamma_{2}^{^{\prime}}}$''
and ``$1-\lambda_{2}\frac{\alpha_{1}}{\gamma_{1}^{^{\prime}}}$''
becomes ``$\lambda_{2}-\lambda_{2}\frac{\alpha_{1}}{\gamma_{1}^{^{\prime}}}$''.
Since $\lambda_{1},\lambda_{2}<1$, the overall multiplier on $\epsilon_{i}$
becomes smaller in magnitude\footnote{Note that $\alpha_{1}/\gamma_{1}^{^{\prime}}\leq1$ and $\alpha_{2}/\gamma_{2}^{^{\prime}}\leq1$
by equation \eqref{eq:gamma1_deri}.}. Essentially, by using the estimated conditional expected output
$\tilde{y}_{i}$, the raw ``1'' share of $\epsilon_{i}$ in $y_{i}$
is moved into the first-stage estimation error of $\overline{y}_{i}$,
which is then ``averaged'' and reduced in magnitude to $\lambda_{1}$
or $\lambda_{2}$, thus leading to smaller overall variance.

Lastly, we emphasize that the efficiency comparison in \ref{thm:StarBetter}
does not directly relate to the theory of semiparametric efficiency
bounds, such as in \citet{ackerberg2014asymptotic}, which is about
asymptotic efficiency of semiparametric estimators under a given criterion
function. In fact, by \citet{ackerberg2014asymptotic}, both estimators
based on $y_{i}$ and $\tilde{y}_{i}$ attain their corresponding
semiparametric efficiency bounds with respect to their\ different
criterion functions $g$ and $g^{*}$. Theorem \ref{thm:StarBetter},
however, is a comparison across the two criterion functions $g$ and
$g^{*}$: it essentially states that the asymptotically efficient
estimator under $g^{*}$ is even more efficient than the efficient
estimator under $g$.

\subsubsection{Lower-Level Regularity Conditions for Kernel First Step}

Finally, we present a set of lower-level conditions that replace Assumptions
\ref{assu:SP_RegCon} and \ref{assu:AsymLin}, when we use the canonical
Nadaraya-Watson kernel estimator for the nonparametric regression
in Step 1. We emphasize that this subsection simply serves as an illustration
of Assumptions \ref{assu:SP_RegCon}-\ref{assu:AsymLin} and Theorem
\ref{thm:AsymDist}, as our method does not require the use of a specific
form of first-step nonparametric estimators. For sieve (series) first-step
estimators, similar results can be derived based on, for example,
\citet*{newey1994asymptotic}, \citet*{chen2007large} and \citet*{chen2015sieve}.
\begin{assumption}[Example of Lower-Level Conditions with Kernel First Step]
 \label{assu:KernelConds}Let $N_{k}:=\sum_{i=1}^{N}\mathbf{\mathbbm1}\left\{ d_{i}=k\right\} $
denote the number of firms for which $h_{ik}$ is observed, and let
$\hat{\gamma}_{k}$ be the Nadaraya-Watson kernel estimator of $\gamma_{k}$
defined by
\[
\hat{\gamma}_{k}\left(v\right):=\frac{\frac{1}{N_{k}b^{3}}\sum_{d_{i}=k}K\left(\frac{v-v_{ik}}{b^{3}}\right)y_{i}}{\frac{1}{N_{k}b^{3}}\sum_{d_{i}=k}K\left(\frac{v-v_{ik}}{b^{3}}\right)}
\]
where $v_{ik}:=\left(x_{ik},z_{i1},z_{i2}\right)$ for all $i$ such
that $d_{i}=k$. Suppose the following conditions:
\begin{itemize}
\item[(i)] $\lambda_{1}\left(x_{i},z_{i}\right)\in\left(\epsilon,1-\epsilon\right)$
for all $\left(x_{i},z_{i}\right)$ for some $\epsilon>0$.
\item[(ii)] $\left(x_{i},z_{i}\right)$ has compact support in $\mathbb{R}^{4}$
with joint density $f$ that is uniformly bounded both above and below
away from zero.
\item[(iii)] $\mathbb{E}\left[y_{i}^{4}\right]<\infty$ and $\mathbb{E}\left[\left.y_{i}^{4}\right|x_{i},z_{i}\right]f\left(x_{i},z_{i}\right)$
is bounded.
\item[(iv)] $\gamma_{k}$ has uniformly bounded derivatives up to order $p\geq4$.
\item[(v)] $K\left(u\right)$ has uniformly bounded derivatives up to order
$p$, $K\left(u\right)$ is zero outside a bounded set, $\int K\left(u\right)du=1$,
$\int u^{t}K\left(u\right)du=\mathbf{0}$ for $t=1,...,p-1$, and
$\int\left\Vert u\right\Vert ^{p}\left|K\left(u\right)\right|du<\infty$.
\item[(vi)] $b$ is chosen such that $\frac{\sqrt{\log N}}{\sqrt{Nb^{3}}}=o\left(N^{-\frac{1}{4}}\right)$
and $\sqrt{N}b^{p}\to0$.
\end{itemize}
\end{assumption}
Assumption \ref{assu:KernelConds}(i) essentially requires that the
proportion of observations with $x_{i1}$ observed and that with $x_{i2}$
observed are both strictly positive, or in other words, the numbers
of both types of observations tend to infinity at the same rate of
$N$. This guarantees that we can estimate both $\gamma_{1}$ based
on observations with $x_{i1}$ and $\gamma_{2}$ based on observations
with $x_{i2}$ well enough asymptotically. Assumption \ref{assu:KernelConds}(iv)
is the key smoothness condition that will help establish the Donsker
property (and a consequent stochastic equicontinuity condition) in
Assumption \ref{assu:SP_RegCon}(i). Assumption \ref{assu:KernelConds}(v)(vi)
are concerned with the choice of kernel function $K$ and bandwidth
parameter $b$: (v) requires that a ``high-order'' kernel function
(of order $p$) is used, while (vi) requires that the bandwidth is
set (with ``under-smoothing'' ) so that the kernel estimator $\text{\ensuremath{\hat{\gamma}_{k}}}$
converges at a rate faster than $N^{-1/4}$, as required in Assumption
\ref{assu:SP_RegCon}(ii). The requirement of $p\geq4$ in (iii) ensures
that (vi) is feasible. Together with the additional regularity conditions
in (ii)(ii), these conditions ensure that Assumptions \ref{assu:SP_RegCon}-\ref{assu:AsymLin}
are satisfied. See \citet*[Section 8.3]{newey1994large} for additional
details.
\begin{thm}[Asymptotic Distributions with Kernel First Step]
\label{thm:Kernel} Under Assumptions \ref{assu:OneLatent}-\ref{assu:StrongMono}
and \ref{assu:KernelConds}, the conclusions of Theorem \ref{thm:AsymDist}
hold.
\end{thm}

\subsection{{\normalsize{}\label{subsec:Discussion} }Generalizations}

\subsubsection*{Additional Instrumental Variables}

If additional instruments are available, it is straightforward to
incorporate them in the second-stage regression, which will take the
form of a two-stage least square estimator instead of an IV regression.
Our results will carry over with suitable changes in notation. For
example, the asymptotic variance formula for $\hat{\alpha}$ needs
to be adapted as
\begin{align*}
\Sigma & :=\left(\Sigma_{xz}\Sigma_{zz}^{-1}\Sigma_{zx}\right)^{-1}\Sigma_{xz}\Sigma_{zz}^{-1}\Omega\Sigma_{zz}^{-1}\Sigma_{zx}\left(\Sigma_{xz}\Sigma_{zz}^{-1}\Sigma_{zx}\right)^{-1}.
\end{align*}

\subsubsection*{Other Parametric Production Functions}

Consider a potentially nonlinear parametric production function of
the form
\[
y_{i}=F_{\alpha}\left(x_{i1},x_{i2}\right)+u_{i}+\epsilon_{i}
\]
After the identification of partially latent inputs via Theorem \ref{thm:ID},
the second stage boils down to the estimation of $\alpha$ based on
the moment condition $\mathbb{E}\left[z_{i}\left(y_{i}-F_{\alpha}\left(x_{i1},x_{i2}\right)\right)\right]=\mathbf{0}$,
which can be obtained via GMM estimation. Technically, since GMM estimators
are Z-estimators, the corresponding asymptotic theory in \citet*{newey1994large},
on which the proof of Theorem \ref{thm:Kernel} is mainly based, still
applies with proper changes in notation.

\subsubsection*{Nonparametric Production Functions}

More generally, with any nonparametric production function that is
additively separable in $u_{i}$ and $\epsilon_{i}$ of the form
\[
y_{i}=F\left(x_{i1},x_{i2}\right)+u_{i}+\epsilon_{i},
\]
where $F$ is an unknown function that satisfies Assumption \ref{assu:MonoF},
the only thing that changes is the second-stage nonparametric estimation
of $F$ with the imputed inputs $\tilde{x}_{i}$ (or more precisely,
with one component known and one component imputed) based on the moment
condition $\mathbb{E}\left[z_{i}\left(y_{i}-F\left(x_{i1},x_{i2}\right)\right)\right]=\mathbf{0}.$
The asymptotic theory for this case can be similarly obtained based
on theory on nonparametric two-step estimation (e.g. \citealp*{ai2007estimation},
and \citealp*{hahn2018nonparametric}).

In the more general specification \eqref{eq:Gen_model}:
\[
y_{i}=F\left(x_{i1},x_{i2},u_{i}\right)+\epsilon_{i}
\]
where there is no more additive separability in $u_{i}$, one way
to obtain identification and implement IV estimation is by adapting
\citet*{chernozhukov2007instrumental} to our current context. Essentially,
we would need to impose strict monotonicity of $F$ in $u_{i}$, impose
independence of $u_{i}$ from $z_{i}$, normalize the distribution
of $u_{i}$ to be uniform, and then exploit a quantile-based residual
condition as described in \citet*{chernozhukov2007instrumental}.

\section{{\normalsize{}\label{sec:MonteCarlo}}A Monte Carlo Experiment}

Here we report the findings of some Monte Carlo experiments. Table
\ref{ParamSpec} reports the parameter specifications of the Cobb-Douglas
production function that we use in our experiments. We assume that
inputs are optimally chosen by a profit maximizing firm as discussed
in detail in Appendix \ref{subsec:CD_Optimal}. Specification 1 is
the baseline specification. These parameters were chosen so that the
simulated data are broadly consistent with the descriptive statistics
of our application that we discuss in detail in Section \ref{sec:FirstAppl}.
Specification 2 has a larger variance in productivity shocks ($u_{i}$).
Specification 3 has a smaller variance in wage distributions ($z_{i}$).
Finally, specification 4 has a larger variance in measurement errors
($e_{i}$).

Specifically, we consider $L$ different markets, with each market
containing $I$ firms, so that the total number of firms is $N=L\times I$.
Firms in the same market $l$ all pay the same local wages, which
we use as the instrumental variables. Local wages are drawn from a
joint log-normal distribution with mean $\mu_{z}$ and variance $\sigma_{z}$
where wages for the two inputs are positively correlated. The firm-level
idiosyncratic productivity shocks and the measurement errors are independently
drawn from normal distributions with zero means and standard deviations $\sigma_{u}$
and $\sigma_{e}$, respectively. The productivity shocks and the measurement
errors are independently and identically distributed both across firms
and markets. We consider different configurations of $L$ and $I$:
specifically, $L=$ 50, 100, 500 and $I=$ 1, 50, 100.

For each experiment, we compute the difference between the true parameter
value and the sample average of the estimates using $M=1000$ replications,
which is a measure of the bias of our estimator. We also estimate
the root mean squared error (RMSE) using the sample standard deviation
of our estimates.

\begin{table}[htbp]
\caption{Monte Carlo Parameter Specification}
\label{ParamSpec}\centering %
\begin{tabular}{cccccccc}
\toprule
 & \multicolumn{4}{c}{Constant Across Specification} & \multicolumn{3}{c}{Variable Across Specification}\tabularnewline
\cmidrule(lr){2-5} \cmidrule(lr){6-8}  & $\alpha_{0}$  & $\alpha_{1}$  & $\alpha_{2}$  & $\mu_{z}$  & $\sigma_{z}$  & $\sigma_{u}$  & $\sigma_{\epsilon}$ \tabularnewline
\midrule
Spec 1  & 4  & 0.35  & 0.25  & $\begin{pmatrix}2.4\\
2.1
\end{pmatrix}$  & $\begin{pmatrix}0.05 & 0.01\\
. & 0.02
\end{pmatrix}$  & 0.4  & 0.3 \tabularnewline
 &  &  &  &  &  &  & \tabularnewline
Spec 2  & 4  & 0.35  & 0.25  & $\begin{pmatrix}2.4\\
2.1
\end{pmatrix}$  & $\begin{pmatrix}0.05 & 0.01\\
. & 0.02
\end{pmatrix}$  & \textbf{{0.8} }  & 0.3 \tabularnewline
 &  &  &  &  &  &  & \tabularnewline
Spec 3  & 4  & 0.35  & 0.25  & $\begin{pmatrix}2.4\\
2.1
\end{pmatrix}$  & $\begin{pmatrix}\boldsymbol{0.02} & 0.01\\
. & 0.02
\end{pmatrix}$  & 0.4  & 0.3 \tabularnewline
 &  &  &  &  &  &  & \tabularnewline
Spec 4  & 4  & 0.35  & 0.25  & $\begin{pmatrix}2.4\\
2.1
\end{pmatrix}$  & $\begin{pmatrix}0.05 & 0.01\\
. & 0.02
\end{pmatrix}$  & 0.4  & \textbf{{0.5} }\tabularnewline
\bottomrule
\end{tabular}
\end{table}

Note that our data generating process mechanically implies $x_{i1}$
and $x_{i2}$ have a linear relationship with $y_{i}$. We estimate
$\gamma_{1}\left(\cd,z_{i}\right)$ and $\gamma_{2}\left(\cd,z_{i}\right)$
using second degree polynomials. Not surprisingly, we find that the
estimated coefficients on quadratic terms are almost 0. The interpolated
functions $\gamma_{1}^{-1}$ and $\gamma_{2}^{-1}$ are also almost
linear.

\begin{table}[htbp]
\caption{Monte Carlo: Different Markets}
\label{MonteCarloResult}\centering %
\begin{tabular}{cccccccc}
\toprule
 & Number of  & Number of  &  & \multicolumn{2}{c}{TSLS} & \multicolumn{2}{c}{Matched TSLS}\tabularnewline
\cmidrule(lr){5-6} \cmidrule(lr){7-8} Param  & Markets (L)  & Firms (I)  & Spec  & Bias  & RMSE  & Bias  & RMSE \tabularnewline
\midrule
$\alpha_{0}$  & 50  & 50  & 1  & 0.001  & 0.000  & 0.001  & 0.000 \tabularnewline
$\alpha_{0}$  & 100  & 100  & 1  & 0.000  & 0.000  & -0.000  & 0.000 \tabularnewline
$\alpha_{0}$  & 50  & 50  & 2  & 0.001  & 0.000  & 0.001  & 0.000 \tabularnewline
$\alpha_{0}$  & 100  & 100  & 2  & 0.000  & 0.000  & -0.000  & 0.000 \tabularnewline
$\alpha_{0}$  & 50  & 50  & 3  & 0.001  & 0.000  & 0.001  & 0.000 \tabularnewline
$\alpha_{0}$  & 100  & 100  & 3  & 0.000  & 0.000  & -0.000  & 0.000 \tabularnewline
$\alpha_{0}$  & 50  & 50  & 4  & 0.001  & 0.000  & 0.001  & 0.000 \tabularnewline
$\alpha_{0}$  & 100  & 100  & 4  & 0.000  & 0.000  & 0.000  & 0.000 \tabularnewline
$\alpha_{0}$  & 500  & 1  & 1  & -0.000  & 0.001  & -0.002  & 0.001 \tabularnewline
$\alpha_{0}$  & 500  & 1  & 2  & -0.001  & 0.002  & -0.002  & 0.002 \tabularnewline
$\alpha_{0}$  & 500  & 1  & 3  & -0.000  & 0.001  & -0.001  & 0.001 \tabularnewline
$\alpha_{0}$  & 500  & 1  & 4  & -0.001  & 0.001  & -0.002  & 0.001 \tabularnewline
\midrule
$\alpha_{1}$  & 50  & 50  & 1  & -0.003  & 0.002  & -0.004  & 0.003 \tabularnewline
$\alpha_{1}$  & 100  & 100  & 1  & 0.000  & 0.001  & 0.001  & 0.001 \tabularnewline
$\alpha_{1}$  & 50  & 50  & 2  & -0.004  & 0.007  & -0.007  & 0.009 \tabularnewline
$\alpha_{1}$  & 100  & 100  & 2  & 0.001  & 0.002  & 0.001  & 0.002 \tabularnewline
$\alpha_{1}$  & 50  & 50  & 3  & -0.005  & 0.006  & -0.007  & 0.007 \tabularnewline
$\alpha_{1}$  & 100  & 100  & 3  & 0.000  & 0.001  & 0.001  & 0.002 \tabularnewline
$\alpha_{1}$  & 50  & 50  & 4  & -0.003  & 0.004  & -0.004  & 0.005 \tabularnewline
$\alpha_{1}$  & 100  & 100  & 4  & 0.000  & 0.001  & 0.001  & 0.001 \tabularnewline
$\alpha_{1}$  & 500  & 1  & 1  & 0.008  & 0.011  & 0.010  & 0.013 \tabularnewline
$\alpha_{1}$  & 500  & 1  & 2  & 0.020  & 0.039  & 0.024  & 0.045 \tabularnewline
$\alpha_{1}$  & 500  & 1  & 3  & 0.007  & 0.027  & 0.010  & 0.033 \tabularnewline
$\alpha_{1}$  & 500  & 1  & 4  & 0.009  & 0.018  & 0.015  & 0.024 \tabularnewline
\midrule
$\alpha_{2}$  & 50  & 50  & 1  & 0.003  & 0.003  & 0.005  & 0.004 \tabularnewline
$\alpha_{2}$  & 100  & 100  & 1  & -0.000  & 0.001  & -0.001  & 0.001 \tabularnewline
$\alpha_{2}$  & 50  & 50  & 2  & 0.003  & 0.010  & 0.006  & 0.012 \tabularnewline
$\alpha_{2}$  & 100  & 100  & 2  & -0.002  & 0.002  & -0.002  & 0.003 \tabularnewline
$\alpha_{2}$  & 50  & 50  & 3  & 0.005  & 0.006  & 0.007  & 0.007 \tabularnewline
$\alpha_{2}$  & 100  & 100  & 3  & -0.000  & 0.001  & -0.001  & 0.002 \tabularnewline
$\alpha_{2}$  & 50  & 50  & 4  & 0.004  & 0.005  & 0.005  & 0.007 \tabularnewline
$\alpha_{2}$  & 100  & 100  & 4  & -0.000  & 0.001  & -0.001  & 0.002 \tabularnewline
$\alpha_{2}$  & 500  & 1  & 1  & -0.013  & 0.017  & -0.015  & 0.019 \tabularnewline
$\alpha_{2}$  & 500  & 1  & 2  & -0.039  & 0.062  & -0.043  & 0.074 \tabularnewline
$\alpha_{2}$  & 500  & 1  & 3  & -0.011  & 0.029  & -0.014  & 0.034 \tabularnewline
$\alpha_{2}$  & 500  & 1  & 4  & -0.014  & 0.026  & -0.020  & 0.036 \tabularnewline
\bottomrule
\end{tabular}
\end{table}

Table \ref{MonteCarloResult} summarizes the performance of two different
estimators: the two stage least squares estimator (TSLS) when all
inputs are observed as well as our first version of TSLS when inputs
are imputed and the observed output is used as the dependent variable.\footnote{Our second version of TSLS, which uses the expected output as the dependent variable, performs almost identically to the first version in this Monte Carlo simulation.}
We refer to this version of the TSLS estimator as the ``matched'' TSLS estimator.
As we would expect given our asymptotic results, the matched TSLS
performs almost as well as the standard TSLS estimator under these
ideal sampling conditions. This finding holds for all four different
specifications and several choices for the number of firms within
a market and the number of local markets.

Next, we investigate how our estimator performs when we have a relatively
small number of observations in each market. Considering an extreme
case, we simulate data for $L=500$ and $I=1$. In this case, as we
only have a single firm in each market, we cannot impute the missing
input variable using within market information. Instead, we pool observations
across markets and estimate conditional expectations conditional on
$x_{1}$ (or $x_{2}$), $z_{1}$, and $z_{2}$.\footnote{Note that the missing inputs are imputed for each market separately
when $I\neq1$.} Table \ref{MonteCarloResult} also summarizes the bias and RMSE where
$L=500$ and $I=1$. We find that the matched TSLS estimator performs
almost as well as the standard TSLS estimator that assumes that both
inputs are observed. The only case where the matched TSLS estimator
exhibits relatively large bias and RMSE is when the variance of the
measurement errors is large (Spec 4). In Appendix \ref{sec:MCPartialWage},
we present a Monte Carlo experiment where we have partially latent
wages.

We have not conducted a Monte Carlo analysis to evaluate how
sensitive our estimator is to other misspecification errors. For
example, it might be interesting to evaluate the performance of our
estimator if input decisions (and hence output) are functions of both
demand and supply shocks.  We conjecture that the magnitude of the
bias that will result from ignoring demand shocks will likely depend
on the correlation structure of the demand and supply shocks and the
market structure of the industry.

\section{{\normalsize{}{}\label{sec:FirstAppl}} Team Production Functions}

Our main application focuses on the estimation of team production
functions. We introduce our data set and discuss our main empirical
findings. Finally, we discuss other data sets that have a similar
structure than the one we use and potential applications of our methods.

\subsection{Institutional Background and Data}

We study team production functions in pharmacies. This industry has
undergone a dramatic change over the past decades. An industry that
used to be primarily dominated by local independent pharmacies has
been transformed by the entry of large chains that operate in multiple
markets.

According to Goldin and Katz (2016) one important technological change in the pharmacy sector
``\emph{is the extensive use of information technology systems
and an increase in prescription drug insurance, which have both
enhanced the ability of pharmacists to hand off clients. Improvements in
information technology have enhanced the ability of pharmacists to leave a
coherent and comprehensive record of each client, increasing the substitutability
of pharmacists and reducing consumer preferences for particular
pharmacists. Because of the increase in insurance coverage, pharmacists
can access the prescriptions of clients through Pharmacy Benefit Managers
(PBM) even if the scripts were not filled at that pharmacy... [Another] change is the  standardization of pharmacy products and services. Medications have been increasingly produced by pharmaceutical
companies rather than being compounded in individual pharmacies and
hospitals. The greater standardization of medications has meant that the
idiosyncratic expertise and talents of a particular pharmacist have become
less important ...  [these] changes make pharmacists better substitutes for each other and enable
an almost costless handoff of clients.}" (p. 708-709)

An important question is the extent to which this transformation
has been driven by technological change that has benefited large chains
over smaller independently operated pharmacies. If this is in fact
the case, these technological changes may help to explain why this
profession has become so popular with females \citep{goldin2016most}.

The main data set that we use is the National Pharmacist Workforce
Survey of 2000 which is collected by Midwestern Pharmacy Research.
The data comes from a cross-sectional survey answered by randomly
selected individual pharmacists with active licenses. The data set
is composed of two types of information: information about pharmacists
and information about the pharmacy each pharmacist works at.

Information at the pharmacy level includes the type of pharmacy \textit{(Independent
}or\textit{\ Chain)}, the hours of operation per week, the number
of pharmacists employed, and the typical number of prescriptions dispensed
at the pharmacies per week. The store-level information is provided
by an individual pharmacist who works at the pharmacy, thus the quality
of the responses may depend on how knowledgeable the person is about
the pharmacy. However, considering that most of the pharmacists in
our sample are observed to be full-time pharmacists, the quality of
the firm-level data is likely to be high. The number of prescriptions
dispensed at the pharmacy is our measure of output. As a consequence,
we do not have to use revenue based output measures which could bias
our analysis as discussed, for example, in \citet*{epple2010new}.

\begin{table}[htb]
\caption{Summary Statistics at the Firm Level: Pharmacies}
\label{tab:ss_pharmacy}\centering \begin{threeparttable}[b] \centering
\begingroup 
\begin{tabular}{lccccccc}
\toprule
Firm  & Number  & Emp  & Operating  & Prescriptions  & Prescriptions  & Prop  & Number \tabularnewline
Type  & Pharmacists  & Size  & Hours  & per Week  & per Hour  & Urban  & of Obs \tabularnewline
\midrule
Indep  & $n<2$  & 3.15  & 51.96  & 778.00  & 14.94  & 0.63  & 50 \tabularnewline
 &  & (1.41)  & (7.08)  & (368.95)  & (6.54)  & (0.39)  & \tabularnewline
Indep  & $2\leq n<3$  & 3.94  & 56.99  & 914.40  & 16.09  & 0.71  & 58 \tabularnewline
 &  & (1.80)  & (10.04)  & (472.81)  & (8.43)  & (0.34)  & \tabularnewline
Indep  & $3\leq n$  & 4.71  & 64.24  & 1252.22  & 19.44  & 0.78  & 36 \tabularnewline
 &  & (1.44)  & (14.15)  & (610.61)  & (8.75)  & (0.32)  & \tabularnewline
\midrule
Chain  & $n<2$  & 1.88  & 53.50  & 666.88  & 12.90  & 0.81  & 8 \tabularnewline
 &  & (0.99)  & (8.02)  & (278.84)  & (6.58)  & (0.34)  & \tabularnewline
Chain  & $2\leq n<3$  & 3.25  & 80.50  & 1294.68  & 16.21  & 0.81  & 101 \tabularnewline
 &  & (1.36)  & (9.86)  & (595.08)  & (7.66)  & (0.29)  & \tabularnewline
Chain  & $3\leq n$  & 5.32  & 82.82  & 1765.67  & 21.43  & 0.89  & 79 \tabularnewline
 &  & (1.63)  & (13.67)  & (681.57)  & (7.87)  & (0.20)  & \tabularnewline
\bottomrule
\end{tabular}\endgroup \begin{tablenotes}

\item Independent pharmacies: fewer than 10 stores under the same
ownership.

\item Chain pharmacies: more than 10 stores under the same ownership.

\item Standard deviations in parentheses.

\item One part-time pharmacist is counted as 0.5 pharmacist in number
of pharmacists.

\item Employment size includes interns and technicians. \end{tablenotes}
\end{threeparttable}
\end{table}

Table \ref{tab:ss_pharmacy} summarizes the means of key variables
that are observed at the firm or pharmacy level. After eliminating
cases with missing input/output information, we observe 332 pharmacists.
Table \ref{tab:ss_pharmacy} suggests that there are some pronounced
differences between chains and independent pharmacies. Chains are
more likely to be located in larger urban areas than independent pharmacies.
They also operate longer hours per week. Interestingly, hourly production
measured by the number of prescriptions per hour is, on average, similar
to the independent pharmacies with similar employment size.\footnote{Most pharmacies in our sample have one manager pharmacist and one
employee pharmacist, but there are a few pharmacies with a larger
employment size.} We explore these issues in more detail below and test whether the
different types of pharmacies have access to the same technology.

\begin{table}[tbph]
\caption{Summary Statistics at the Worker Level: Pharmacists}
\label{tab:ss_individual} \centering \begin{threeparttable}[b]
\begingroup %
\begin{tabular}{llccccc}
\toprule
Firm  &  & Number of  & Actual  & Paid  & Hourly  & Number of \tabularnewline
Type  & Position  & Pharmacists  & Hours  & Hours  & Earnings  & Obs \tabularnewline
\midrule
Indep  & Employee  & $n<2$  & 40.94  & 39.28  & 28.87  & 9 \tabularnewline
 &  &  & (11.61)  & ( 9.60)  & (7.64)  & \tabularnewline
Indep  & Employee  & $2\leq n<3$  & 33.90  & 33.03  & 29.37  & 29 \tabularnewline
 &  &  & (12.01)  & (11.14)  & (4.09)  & \tabularnewline
Indep  & Employee  & $3\leq n$  & 31.61  & 30.95  & 30.24  & 28 \tabularnewline
 &  &  & (11.62)  & (10.96)  & (4.93)  & \tabularnewline
Indep  & Manager  & $n<2$  & 50.02  & 45.34  & 30.32  & 41 \tabularnewline
 &  &  & (9.05)  & (7.24)  & (12.45)  & \tabularnewline
Indep  & Manager  & $2\leq n<3$  & 49.45  & 44.19  & 28.70  & 29 \tabularnewline
 &  &  & (8.15)  & (7.99)  & (9.90)  & \tabularnewline
Indep  & Manager  & $3\leq n$  & 46.50  & 44.38  & 30.28  & 8 \tabularnewline
 &  &  & (4.11)  & (6.30)  & (6.57)  & \tabularnewline
\midrule
Chain  & Employee  & $n<2$  & 46.20  & 43.00  & 34.70  & 5 \tabularnewline
 &  &  & (2.77)  & (4.47)  & (2.19)  & \tabularnewline
Chain  & Employee  & $2\leq n<3$  & 41.82  & 39.84  & 34.13  & 66 \tabularnewline
 &  &  & (5.76)  & (4.38)  & (3.32)  & \tabularnewline
Chain  & Employee  & $3\leq n$  & 39.96  & 37.94  & 34.03  & 56 \tabularnewline
 &  &  & (8.63)  & (7.02)  & (3.12)  & \tabularnewline
Chain  & Manager  & $n<2$  & 45.33  & 42.00  & 36.75  & 3 \tabularnewline
 &  &  & (5.03)  & (2.65)  & (4.43)  & \tabularnewline
Chain  & Manager  & $2\leq n<3$  & 44.10  & 40.50  & 34.06  & 35 \tabularnewline
 &  &  & (7.02)  & (2.58)  & (4.90)  & \tabularnewline
Chain  & Manager  & $3\leq n$  & 43.61  & 41.43  & 35.04  & 23 \tabularnewline
 &  &  & (5.41)  & (3.41)  & (3.59)  & \tabularnewline
\bottomrule
\end{tabular}\endgroup \begin{tablenotes}

\item Independent pharmacies: fewer than 10 stores under the same
ownership.

\item Chain pharmacies: more than 10 stores under the same ownership.

\item Hourly earnings are computed based on the paid hours, not actual
hours.

\item Standard deviations in parentheses. \end{tablenotes} \end{threeparttable}
\end{table}

The survey also collects various information about pharmacists including
hours of work, demographics, and household characteristics. Most importantly
we observe the position at the pharmacy \textit{(Owner/Manager} or
\textit{\ Employee)}. We treat hours of the manager and hours of
the employees as the two input factors in the production function.

Information related to the individual pharmacists is summarized in
Table \ref{tab:ss_individual}. Employee pharmacists at independent
pharmacies work fewer hours than the employee pharmacists at chain
pharmacies, and hourly earnings are lower than those of the employees
at the chains. Pharmacists in managerial positions at independent
pharmacies work more hours than managers at chain pharmacies, but
they have lower hourly earnings on average.

\subsection{Empirical Results}

We observe only one pharmacy in each local labor market, which is
defined as the 5-digit zip code area. Hence, we use the version of
our estimator that averages across local markets as discussed in Section
\ref{sec:MonteCarlo}. We use imputed hourly wages for the manager
and the employer as the additional observed covariates ($z_{1},z_{2}$)
for the first-stage estimation.\footnote{In this application, we only observe the wage for the observed type.
Thus, wages are imputed using local demand shifters in 5-digit zip
code levels and pharmacists' characteristics. We have verified with
additional Monte Carlo simulation exercise reported in Appendix C
that our method performs as well as the standard TSLS estimator with
this variation.} In the second stage estimation, we use the observed wage for the
observed position, the imputed wages, and principal components of local demand shifters
as additional instruments.\footnote{The local demand shifters include total population size, median household
income, and proportion of households with retirement income.}

We implement two versions of our ``matched'' TSLS estimator: the
first estimator uses the observed outputs while the second one uses
expected outputs. Since the observed output is subject to a measurement
error, the semi-parametric estimator using expected outputs offers
the potential of some efficiency gains as discussed in Theorem \ref{thm:StarBetter}.
Table \ref{tab:est1} summarizes our findings. We report the estimated
parameters of the Cobb-Douglas production function as well as the
estimated standard errors. In addition, we report standard F-statistics
for the first stage of the TSLS estimator to test for weak instruments.
Overall, we find that our instruments are sufficiently strong in most
cases.\footnote{As a robustness check, we also explored a different matching algorithm
which estimates the expectation of output conditional on local demand
shifters rather than wages. The results are consistent although the
matching algorithm with local demand shifters gives slightly larger
point estimates with slightly less precision.}

Table \ref{tab:est1} shows that we estimate most of parameters of
the production function with good precision. Correcting for potential
measurement error by using the expected output as the dependent variable,
we achieve similar, maybe even slightly more plausible estimates.

\begin{table}[htb]
\caption{Estimation Result}
\label{tab:est1}\centering %
\begin{tabular}{ccccc}
\toprule
 & \multicolumn{2}{c}{Independent} & \multicolumn{2}{c}{Chain}\tabularnewline
 & Observed  & Expected  & Observed  & Expected \tabularnewline
 & Outputs  & Outputs  & Outputs  & Outputs \tabularnewline
\midrule
$\alpha_{0}$  & 5.447  & 5.857  & 2.504  & 3.634 \tabularnewline
 & (0.597)  & (0.331)  & (1.790)  & (1.060) \tabularnewline
$\alpha_{1}$  & 0.227  & 0.163  & 0.819  & 0.687 \tabularnewline
 & (0.122)  & (0.057)  & (0.454)  & (0.268) \tabularnewline
$\alpha_{2}$  & 0.090  & 0.047  & 0.409  & 0.250 \tabularnewline
 & (0.071)  & (0.051)  & (0.191)  & (0.105) \tabularnewline
\midrule
Nobs  & 144  & 144  & 188  & 188 \tabularnewline
First-stage F for $x_{1}$  & 9.320  & 9.320  & 11.774  & 11.774 \tabularnewline
First-stage F for $x_{2}$  & 13.648  & 13.648  & 3.630  & 3.630 \tabularnewline
\bottomrule
\end{tabular}
\end{table}

Our results provide several insights into understanding the difference
between independents and chains. First, our results indicate that
chains may have a different production function than independent pharmacies.
A formal joint hypothesis test reported in Table \ref{tab:test1}
rejects the null hypothesis that the coefficients of the production
function are the same.

Second, our findings also suggest that managers may be more effective
in chains than independents. A formal one-sided t-test reported in
Table \ref{tab:test1} rejects the null hypothesis that the two coefficients
that characterize the marginal product of inputs are the same. Our findings imply that larger chains are more efficient in many other dimensions, such as task assignments or logistics, so that managers at larger chains may be able to focus more on prescription-related tasks compared to managers at independent pharmacies.

\begin{table}[htb]
\caption{Hypothesis Tests}
\label{tab:test1}\centering %
\begin{tabular}{lccc}
\toprule
 & Production Function  & Marginal Product  & Residual Variance \tabularnewline
 & (Joint)  & of Manager $\alpha_{1}$  & $V(u)$ \tabularnewline
\midrule
Independent  &  & 0.163  & 0.010 \tabularnewline
Chain  &  & 0.687  & 0.006 \tabularnewline
Difference or Ratio  &  & -0.524  & 1.532 \tabularnewline
\midrule
Test Statistics  & 122.841  & -1.913  & 1.532 \tabularnewline
Test  & \textit{Wald}  & \textit{t}  & \textit{F} \tabularnewline
p-value  & (0.000)  & (0.028)  & (0.003) \tabularnewline
\bottomrule
\end{tabular}
\end{table}

Finally, we find that chains have a significantly lower residual variance
than independents. A formal F test reported in Table \ref{tab:test1}
rejects the null hypothesis that the residual variance of independents
is greater than or equal to the residual variance of chains. Note
that all the tests are based on the estimation results with the expected
outputs as the dependent variable.

We also test whether the observed labor inputs are indeed the optimal
choice of firms. If the inputs are optimally chosen, the coefficients
can be directly estimated from equation (\ref{eq:CD_RF}) in Appendix
\ref{subsec:CD_Optimal}. Under the assumption of Cobb-Douglas production,
we can test the optimality by jointly testing the null hypothesis
of equality of both coefficients. Table \ref{tab:optimality_test}
shows the results. A formal Wald test rejects the null hypothesis
of optimality. Thus, the direct inversion of the optimality conditions
cannot be applied to estimate the parameters of the production function,
whereas our new estimator is feasible. We note, however, that the
test of the optimality of inputs here is built upon the premise that
the linear model is correctly specified and all the assumptions underlying
our method are valid.

\begin{table}[htb]
\caption{Test for Optimality of Inputs}
\label{tab:optimality_test} \centering %
\begin{tabular}{ccccc}
\toprule
 & \multicolumn{2}{c}{Independent} & \multicolumn{2}{c}{Chain}\tabularnewline
 & $x_{1}$ Observed  & $x_{2}$ Observed  & $x_{1}$ Observed  & $x_{2}$ Observed \tabularnewline
\midrule
Wald Statistic  & 5.495  & 36.914  & 15.312  & 26.172 \tabularnewline
p-value  & (0.064)  & (0.000)  & (0.000)  & (0.000) \tabularnewline
\bottomrule
\end{tabular}
\end{table}

Although most pharmacies in our sample have one manager and one pharmacist,
there are a few pharmacies with more than one employee pharmacist.
For this subset of pharmacies, we compute the total hours worked by
employee pharmacists by multiplying the reported hours worked from
an employee by the number of employees. Then, the second imputation
step is applied based on the total hours worked by all employees.
In this process, we implicitly assume the labor hours from two different
employees are perfect substitutes. As a robustness check, we also
estimate a version of production function which has an elasticity
of substitution between the hours worked by different employees equal
to one. Table \ref{tab:est1_robustcheck} summarizes this version
of the estimation result. The estimated parameters show that employees
have slightly lower marginal products at both independents and chains compared
to our baseline estimation, but in general our estimation result is
robust to how we treat employee inputs from pharmacies with more than
one employee.

\begin{table}[htbp]
\caption{Using $N_{2}*log(x_{2})$ instead of $log(N_{2}*H_{2})$}
\label{tab:est1_robustcheck}\centering %
\begin{tabular}{ccccc}
\toprule
 & \multicolumn{2}{c}{Independent} & \multicolumn{2}{c}{Chain}\tabularnewline
 & Observed  & Expected  & Observed  & Expected \tabularnewline
 & Outputs  & Outputs  & Outputs  & Outputs \tabularnewline
\midrule
$\alpha_{0}$  & 5.493  & 5.888  & 3.409  & 4.201 \tabularnewline
 & (0.527)  & (0.270)  & (1.656)  & (0.972) \tabularnewline
$\alpha_{1}$  & 0.258  & 0.178  & 0.878  & 0.719 \tabularnewline
 & (0.121)  & (0.057)  & (0.446)  & (0.261) \tabularnewline
$\alpha_{2}$  & 0.033  & 0.017  & 0.092  & 0.056 \tabularnewline
 & (0.021)  & (0.014)  & (0.039)  & (0.022) \tabularnewline
\midrule
Nobs  & 144  & 144  & 188  & 188 \tabularnewline
First-stage F for $x_{1}$  & 10.066  & 10.066  & 10.199  & 10.199 \tabularnewline
First-stage F for $x_{2}$  & 12.360  & 12.360  & 3.210  & 3.210 \tabularnewline
\bottomrule
\end{tabular}
\end{table}

We thus conclude that chains have different production functions than
independent pharmacies which may partially explain the change in the
observed market structure of that industry.

One main concern with this application is that the variation of
output (the number of prescriptions) in the pharmacy industry is also
driven by demand shocks. We offer four observations.

First, our
framework works as long as there is only one structural shock. Our
model is agnostic about whether $u_i$ is a technology or a demand
shock.  However, the choice of instruments depends on the
interpretation of the structural shock. In particular, wages are less
plausible instruments if the structural shock is primarily a demand
shock.

Second, our approach can also be applied in a model with two
structural shocks as long as the second shock (the demand shock) is
unexpected and iid across firms. The main results of the paper are
still valid since input decisions are by assumption not affected by
these unexpected demand shocks, and we use expected and not realized
output in the matching procedure. Hence, we implicitly assume that
demand shocks can be treated as we treat measurement error, i.e. as
additively separable iid shocks.

Third, it is reasonable to consider a more general model with two
structural shocks and measurement error.  This generalized model can
then account for the possibility that input decisions are functions of
both demand and supply shocks. Moreover, these shocks do not
necessarily enter into the production function in a simple additively
separable way. As we discuss in detail in the conclusions of this
paper, this models falls outside the scope of this paper. More
research is needed to develop plausible identification strategies for
this generalized model.

Finally, there are additional issues that arise when estimating
production functions in retail compared to manufacturing. For example,
capacity constraints, rationing based on wait time, and service
quality may all be important. Unfortunately, we do not observe these
variables in our data set. Hence, we cannot evaluate how important
these factors are. More research
is needed to fully address this important research question.

\subsection{Discussion of Other Applications}

As we discussed above, the partially latent data structure that we
study in this paper arises quite naturally in matched employer-employee
data sets which contain information collected from individuals as
well as information collected from businesses or establishments. Hence,
they contain useful information about the firm such as output and
revenues as well as the employees of the firm. However, the sampling
design often implies that only a small subset of the employees of
a firm are included in the sample. If the survey does not sample all
employees in a firm then some important labor inputs are likely to
be latent from the perspective of the econometrician.\footnote{For an early survey of employer-employee data sets see \citet*{Abowd-Kramaz-99}.}

Consider the Workers Establishment Characteristics Database (WECD)
which matches long form respondents of the 1990 U.S. Census of Population
to data on their employers from the Longitudinal Research Database
(LRT). Note that the LRT not only has detailed output measures, but
also includes measures of the capital stock, which is important in
manufacturing. Since the long-form was only given to a random sample
of approximately five percent of the population, we can typically
only match five percent of the workers for large plants and a much
lower percentage for some smaller plants. For example, \citet*{HNT-99}
focused on large plants with more than 20 employees and rely on the
aggregate data in the LRD to estimate the production function. Hence,
they do not differentiate between different types of labor inputs,
such blue and white collar workers. The methods developed in our paper
allows researchers to use the data from the WECD to estimate the production
function of small plants, even if we do not observe white and blue
collar workers for each small plant.\footnote{The New Worker-Employer Characteristic Dataset is an extension of
the WECD that also contains firms. Another relevant data set is the
Longitudinal Employer-Household Dynamics (LEHD) which is drawn from
state unemployment insurance administrative files, and hence, contains
a larger subset of workers. It suffers from missing data problems,
especially for firms with high turnover in the labor force. Hence,
we often only observe a subset of workers for these firms. This is
particularly problematic for small firms in the service sector.}

Another well-known employer-employee data set that applied this sampling
design is the Workplace and Employee Survey (WES), a longitudinal
survey of workplaces and their employees administered annually by
Statistics Canada between 1999 and 2006. Every year, a representative
sample of approximately 6,000 employers was surveyed. A maximum of
twenty-four employees were randomly interviewed from each sampled
workplace in each odd year and re-interviewed the following year.
Thus, we only observe a small subsample of the employees at each firm
in the WES.\footnote{See, for example, \citet*{Dostie-Javdani-20}, for more details.}

Other applications of our techniques can be based on surveys conducted
by professional associations. Consider, for example, the \textit{After
the JD} , which is a nationally representative longitudinal data set
constructed by the American Bar Foundation and the National Association
for Law Placement. The data set follows students who graduated in
2000. The respondents were surveyed three times: once each in 2003,
2007, and 2012. The survey thus contains detailed information about
associates and partners. The commonly available version of the After
the JD does not contain the identity of the law firm. However, this
information is available to internal researchers at the American Bar
Association. Hence, one should be able to match the employees in the
AJD to individual law firms and study the productivity differences
among law firms using our technique.

Finally, there are numerous applications outside of industrial organization.
In Appendix D show that the partially latent covariate problem can also arise
in the study of intergenerational transmission of human capital, wealth
or attitudes. It is rather common that we do not sample all relevant
household members. Here, we consider, the Child Development Supplement
of the PSID to study the transmission of human capital from parents
to children. If a child grows up in traditional family, it is likely
that we observe inputs from both mother and father. However, traditional
family arrangements have become less common during the past decades
and thus the focus has shifted on evaluating the impact of growing
up in less traditional environments. If a child grows up in a non-traditional
family one of the two parents' inputs are often missing. We can also
apply our framework to study achievement functions. We find that there
are some significant differences between married and divorced parents.
In particular, divorced fathers have no significant impact on child
quality.

\section{{\normalsize{}{}\label{sec:Conclusion}}Concluding Remarks}

We have developed a new method for identifying econometric models
with partially latent covariates. We have shown that a broad class
of econometric models that play a large role in industrial organization
and labor economics can be non-parametrically identified if the partially
latent covariates are monotonic functions of a common shock. Examples
that fall into this class of models are production and skill formation
functions. The partially latent data structure arises quite naturally
in these settings if we employ an ``input-based sampling'' strategy,
i.e. if the sampling unit is one of multiple labor input factors.
It is plausible that the sampling unit will only have incomplete information
about the other labor inputs that affect output. Our proofs of identification
are constructive and imply a sequential, two-step semi-parametric
estimation strategy. We have discussed the key problems encountered
in estimation, characterized rate of convergence, and the asymptotic
distribution of our estimators. Our application focuses on estimating
team production functions. Using a national survey of pharmacists,
we have found some convincing evidence that chains have different
technologies than independently operated pharmacies. In particular,
managers appear to be to have higher marginal products in chains.

Finally, our research provides ample scope for future research in econometric
methodology. We have restricted ourselves to applications in which
our method of identification can be combined with standard IV techniques
to estimate the functions of interest. Much of the recent panel data
literature has focused on dynamic inputs in the presence of adjustment
costs, and more research is needed to extend the idea in this paper
to a fully fledged dynamic panel data framework. We have also restricted
ourselves to systems of inputs with a single common shock. Another
potentially interesting research question is how our methods can be
extended to more complicated econometric structures with multiple
shocks.
\begin{acknowledgement*}
We would like to thank Xu Cheng, Aureo de Paula, Ulrich Doraszelski,
Amit Gandhi, Claudia Goldin, Aviv Nevo, Dan Silverman, Petra Todd,
and seminar participants at numerous universities for comments and
suggestions. Postlewaite and Sieg acknowledge support from the National Science Foundation.
\end{acknowledgement*}

\nocite{npws,psidcds}
\bibliographystyle{ecta}
\bibliography{LatentInputs}

\newpage{}

\appendix

\section{\label{sec:CD}Additional Derivations}

\subsection{{\normalsize{}\label{subsec:CD_Optimal}}Optimal Input Choice in
the Cobb-Douglas Case}

Suppose that firm $i$ chooses inputs optimally by solving the following
(expected) profit-maximization problem:
\begin{equation}
\max_{X_{i1},X_{i2}}e^{\alpha_{0}+u_{i}}X_{i1}^{\alpha_{1}}X_{i2}^{\alpha_{2}}-Z_{i1}X_{i1}-Z_{i2}X_{i2},\label{eq:CD_max}
\end{equation}
where $X_{i1},X_{i2},Z_{i1},Z_{i2}$ denote exponents of $x_{i1},x_{i2},z_{i1},z_{i2}$.
By the first-order conditions,
\begin{align*}
X_{i1} & =e^{\frac{\alpha_{0}+u_{i}}{1-\alpha_{1}-\alpha_{2}}}\left(\frac{Z_{i1}}{\alpha_{1}}\right)^{\frac{1-\alpha_{2}}{\alpha_{1}+\alpha_{2}-1}}\left(\frac{Z_{i2}}{\alpha_{2}}\right)^{\frac{\alpha_{2}}{\alpha_{1}+\alpha_{2}-1}}\\
X_{i2} & =e^{\frac{\alpha_{0}+u_{i}}{1-\alpha_{1}-\alpha_{2}}}\left(\frac{Z_{i2}}{\alpha_{2}}\right)^{\frac{1-\alpha_{1}}{\alpha_{1}+\alpha_{2}-1}}\left(\frac{Z_{i1}}{\alpha_{1}}\right)^{\frac{\alpha_{1}}{\alpha_{1}+\alpha_{2}-1}}\\
\ol Y_{i} & =e^{\frac{\alpha_{0}+u_{i}}{1-\alpha_{1}-\alpha_{2}}}\left(\frac{Z_{i1}}{\alpha_{1}}\right)^{\frac{\alpha_{1}}{\alpha_{1}+\alpha_{2}-1}}\left(\frac{Z_{i2}}{\alpha_{2}}\right)^{\frac{\alpha_{2}}{\alpha_{1}+\alpha_{2}-1}}\\
 & =e^{\alpha_{0}+u_{i}}\left(\frac{\alpha_{2}Z_{i1}}{\alpha_{1}Z_{i2}}\right)^{\alpha_{2}}X_{i1}^{^{\alpha_{1}+\alpha_{2}}}=e^{\alpha_{0}+u_{i}}\left(\frac{\alpha_{1}Z_{i2}}{\alpha_{2}Z_{i1}}\right)^{\alpha_{1}}X_{i2}^{^{\alpha_{1}+\alpha_{2}}}
\end{align*}
In log forms
\begin{align*}
x_{i1} & =h_{1}\left(u_{i},z_{i}\right)=\frac{\alpha_{0}+\left(1-\alpha_{2}\right)\log\alpha_{1}+\alpha_{2}\log\alpha_{2}}{1-\alpha_{1}-\alpha_{2}}-\frac{1-\alpha_{2}}{1-\alpha_{1}-\alpha_{2}}z_{i1}-\frac{\alpha_{2}}{1-\alpha_{1}-\alpha_{2}}z_{i2}+\frac{1}{1-\alpha_{1}-\alpha_{2}}u_{i}\\
x_{i2} & =h_{2}\left(u_{i},z_{i}\right)=\frac{\alpha_{0}+\alpha_{1}\log\alpha_{1}+\left(1-\alpha_{1}\right)\log\alpha_{2}}{1-\alpha_{1}-\alpha_{2}}-\frac{\alpha_{1}}{1-\alpha_{1}-\alpha_{2}}z_{i1}-\frac{1-\alpha_{1}}{1-\alpha_{1}-\alpha_{2}}z_{i2}+\frac{1}{1-\alpha_{1}-\alpha_{2}}u_{i}\\
\ol y_{i} & =\overline{y}\left(u_{i},z_{i}\right)=\frac{\alpha_{0}+\alpha_{1}\log\alpha_{1}+\alpha_{2}\log\alpha_{2}}{1-\alpha_{1}-\alpha_{2}}-\frac{\alpha_{1}}{1-\alpha_{1}-\alpha_{2}}z_{i1}-\frac{\alpha_{2}}{1-\alpha_{1}-\alpha_{2}}z_{i2}+\frac{1}{1-\alpha_{1}-\alpha_{2}}u_{i}\\
 & =\alpha_{0}+\alpha_{2}\log\left(\alpha_{2}/\alpha_{1}\right)+\left(\alpha_{1}+\alpha_{2}\right)h_{1}\left(u_{i},z_{i}\right)+\alpha_{2}z_{i1}-\alpha_{2}z_{i2}+u_{i}\\
 & =\alpha_{0}+\alpha_{1}\log\left(\alpha_{1}/\alpha_{2}\right)+\left(\alpha_{1}+\alpha_{2}\right)h_{2}\left(u_{i},z_{i}\right)-\alpha_{1}z_{i1}+\alpha_{1}z_{i2}+u_{i}
\end{align*}
Taking inverses
\begin{align*}
u_{i} & =h_{1}^{-1}\left(x_{i1},z_{i}\right):=-\left[\alpha_{0}+\left(1-\alpha_{2}\right)\log\alpha_{1}+\alpha_{2}\log\alpha_{2}\right]+\left(1-\alpha_{1}-\alpha_{2}\right)x_{i1}+\left(1-\alpha_{2}\right)z_{i1}+\alpha_{2}z_{i2}\\
 & =h_{2}^{-1}\left(x_{i2},z_{i}\right):=-\left[\alpha_{0}+\alpha_{1}\log\alpha_{1}+\left(1-\alpha_{1}\right)\log\alpha_{2}\right]+\left(1-\alpha_{1}-\alpha_{2}\right)x_{i2}+\alpha_{1}z_{i1}+\left(1-\alpha_{1}\right)z_{i2}
\end{align*}

Hence,
\begin{align*}
\gamma_{1}\left(x_{i1},z_{i}\right) & =\overline{y}\left(h_{1}^{-1}\left(x_{i1},z_{i}\right),z_{i}\right)=-\log\alpha_{1}+x_{i1}+z_{i1},\\
\gamma_{2}\left(x_{i2},z_{i}\right) & =\overline{y}\left(h_{2}^{-1}\left(x_{i2},z_{i}\right),z_{i}\right)=-\log\alpha_{2}+x_{i2}+z_{i2},
\end{align*}

and
\begin{align}
y_{i}=\gamma_{1}\left(x_{i1},z_{i}\right)+\epsilon_{i}= & -\log\alpha_{1}+x_{i1}+z_{i1}+\epsilon_{i}\nonumber \\
=\gamma_{2}\left(x_{i2},z_{i}\right)+\epsilon_{i}= & -\log\alpha_{2}+x_{i2}+z_{i2}+\epsilon_{i}.\label{eq:CD_RF}
\end{align}
It is then evident that $\alpha_{1}$ or $\alpha_{2}$ can be estimated
directly from \eqref{eq:CD_RF} from the corresponding subsample where
$x_{i1}$ or $x_{i2}$ is observed. Furthermore, we may test input
optimality based on equation \eqref{eq:CD_RF}.

\subsection{\label{subsec:CD_StratComp}Nash Equilibrium under Strategic Complementarity}

Suppose that, given $u,z$ and the other partner's choice $X_{2}$,
partner 1 solves
\[
\max_{X_{1}}\pi_{1}\left(X,u;Z\right):=\l_{1}\left(F\left(X,u\right)-Z_{1}X_{1}-Z_{2}X_{2}\right)+Z_{1}X_{1}-\frac{1}{2}c_{1}X_{1}^{2},
\]
where $F\left(X,u\right):=e^{u+\a_{0}}X_{1}^{\a_{1}}X_{2}^{\a_{2}}$,
$\l_{1}\in\left(0,1\right)$ and $c_{1}>0$. Similarly, partner 2
solves
\[
\max_{X_{2}}\pi_{2}\left(X,u;Z\right):=\l_{2}\left(F\left(X,u\right)-Z_{1}X_{1}-Z_{2}X_{2}\right)+Z_{2}X_{2}-\frac{1}{2}c_{2}X_{2}^{2},
\]
with $\l_{2}\in\left(0,1\right)$ and $c_{2}>0$.

Since the game is supermodular by \eqref{eq:eg_strat_comp}, the set
of Nash equilibria admits a minimum and a maximum under the partial
order defined by bivariate monotonicity. Let $X^{*}\left(u,Z\right)$
the minimum NE and $X^{**}$ the maximum NE.

Suppose that $X^{*}\neq X^{**}$. Then WLOG suppose that $X^{*}\lneqq X^{**}$.
Moreover, both $X^{*}$ and $X^{**}$ must solve the FOCs:
\begin{align*}
\Dif_{X_{1}}F\left(X,u\right)+\frac{1-\l_{1}}{\l_{1}}Z_{1}-\frac{c_{1}}{\l_{1}}X_{1} & =0,\\
\Dif_{X_{2}}F\left(X,u\right)+\frac{1-\l_{2}}{\l_{2}}Z_{2}-\frac{c_{2}}{\l_{2}}X_{2} & =0,
\end{align*}
or in vector form
\begin{equation}
\Dif_{X}F\left(X,u\right)=\left(\begin{array}{c}
\frac{c_{1}}{\l_{1}}X_{1}\\
\frac{c_{2}}{\l_{2}}X_{2}
\end{array}\right)-\left(\begin{array}{c}
\frac{1-\l_{1}}{\l_{1}}Z_{1}\\
\frac{1-\l_{2}}{\l_{2}}Z_{2}
\end{array}\right).\label{eq:FOC_vec}
\end{equation}

Taking difference of \eqref{eq:FOC_vec} evaluated at $X^{*}$ and
$X^{**}$, we have
\[
\Dif_{X}F\left(X^{^{**}},u\right)-\Dif_{X}F\left(X^{*},u\right)=\left(\begin{array}{c}
\frac{c_{1}}{\l_{1}}\left(X_{1}^{**}-X_{1}^{*}\right)\\
\frac{c_{2}}{\l_{2}}\left(X_{2}^{**}-X_{2}^{*}\right)
\end{array}\right),
\]
and hence, by $X^{*}\neq X^{**}$,
\begin{align}
 & \left(X^{**}-X^{*}\right)^{'}\left(\Dif_{X}F\left(X^{^{**}},u\right)-\Dif_{X}F\left(X^{*},u\right)\right)\nonumber \\
= & \frac{c_{1}}{\l_{1}}\left(X_{1}^{**}-X_{1}^{*}\right)^{2}+\frac{c_{2}}{\l_{2}}\left(X_{2}^{**}-X_{2}^{*}\right)^{2}>0.\label{eq:pos}
\end{align}
In the meanwhile, we have
\[
\Dif_{X}F\left(X^{^{**}},u\right)-\Dif_{X}F\left(X^{*},u\right)=\Dif_{XX}F\left(\tilde{X},u\right)\left(X^{**}-X^{*}\right)
\]
for some $\tilde{X}$ between $X^{*}$ and $X^{**}$, and we know
$\Dif_{XX}F\left(\tilde{X},u\right)$ must be negative semi-definite
for Cobb-Douglas production functions under the assumption of $\a_{1}+\a_{2}\leq1$,
which implies
\begin{align*}
 & \left(X^{**}-X^{*}\right)^{'}\left(\Dif_{X}F\left(X^{^{**}},u\right)-\Dif_{X}F\left(X^{*},u\right)\right)\\
= & \left(X^{**}-X^{*}\right)^{'}\Dif_{XX}F\left(\tilde{X},u\right)\left(X^{**}-X^{*}\right)\leq0,
\end{align*}
a contradiction to \eqref{eq:pos}. Hence, we conclude that $X^{*}=X^{**}$
and thus the NE is unique.

Take any $\ol u>\ul u$ and write $\ol X:=X^{*}\left(\ol u,Z\right)$
and $\ul X:=X^{*}\left(\ul u,Z\right)$. By, for example, \citet*{milgrom1990rationalizability},
we know that $\ol X\geq\ul X$.

Suppose that $\ol X_{1}=\ul X_{1}$. Then
\begin{align*}
\Dif_{X_{1}}F\left(\ol X_{1},\ol X_{2},\ol u\right) & =\frac{c_{1}}{\l_{1}}\ol X_{1}-\frac{1-\l_{1}}{\l_{1}}Z_{1}\\
 & =\frac{c_{1}}{\l_{1}}\ul X_{1}-\frac{1-\l_{1}}{\l_{1}}Z_{1}=\Dif_{X_{1}}F\left(\ul X,\ul u\right)\\
 & =\Dif_{X_{1}}F\left(\ol X_{1},\ul X_{2},\ul u\right)\\
 & \leq\Dif_{X_{1}}F\left(\ol X_{1},\ol X_{2},\ul u\right)\quad\text{ by \ensuremath{\Dif_{X_{2}X_{1}}F>0}}\text{ in }\eqref{eq:eg_strat_comp}
\end{align*}
which contradicts with \eqref{eq:eg_inc_diff}.

Thus $\ol X_{1}>\ul X_{1}$ and similarly $\ol X_{2}>\ul X_{2}$.

Hence, both $X_{1}^{*}\left(u,Z\right)$ and $X_{2}^{*}\left(u,Z\right)$
must be strictly increasing in $u$, satisfying Assumption \eqref{assu:input_mono}.

\section{{\normalsize{}{}\label{sec:EstProof}}Proofs}

\subsection{Additional Notation and Lemmas}

\paragraph*{Notation}

For each $i$, we use $x_{ij}$ to denote the observed input and use
$x_{ik}$ to denote the latent input variable for firm $i$, i.e.
\begin{align*}
x_{ij} & =x_{i1},\ x_{ik}=x_{i2},\ \text{for }d_{i}=1,\\
x_{ij} & =x_{i2},\ x_{ik}=x_{i1},\ \text{for }d_{i}=2.
\end{align*}
We write
\begin{align*}
d_{i1} & :=\mathbf{\mathbbm1}\left\{ d_{i}=1\right\} ,\\
d_{i2} & :=\mathbf{\mathbbm1}\left\{ d_{i}=2\right\} ,
\end{align*}
so that $x_{ij}=d_{i1}x_{i1}+d_{i2}x_{i2}$ while $x_{ik}:=d_{i1}x_{i2}+d_{i2}x_{i1}$.
We write $\overline{x}_{i}:=\left(1,x_{i1},x_{i2}\right)^{^{\prime}}$
to denote the true regressor vector. (Recall $\tilde{x}_{i}$ denotes
the same regressor vector with imputed latent input $\hat{x}_{ik}$
in place of $x_{ik}$.)

Moreover, we suppress the instrumental variables $z_{i}$ in functions,
such as $\gamma_{1}\left(u_{i},z_{i}\right)$, unless it becomes necessary
to emphasize the dependence of such functions on $z_{i}$.
\begin{lem}
\label{lem:bound_inv}Under Assumption \ref{assu:StrongMono}, if
$\left\Vert \hat{\gamma}_{k}-\gamma_{k}\right\Vert _{\infty}=O_{p}\left(a_{n}\right)$,
then $\left\Vert \hat{\gamma}_{k}^{-1}-\gamma_{k}^{-1}\right\Vert _{\infty}=O_{p}\left(a_{n}\right)$
and $\left|\hat{x}_{ik}-x_{ik}\right|=O_{p}\left(a_{n}\right)$.
\end{lem}
\begin{proof}
By Assumption \ref{assu:StrongMono} we have
\[
\ul c\left|u_{1}-u_{2}\right|\leq\left|\g_{k}\left(u_{1}\right)-\g_{k}\left(u_{2}\right)\right|
\]
For any $v\in Range\left(\g_{k}\right)$,
\begin{align*}
\left|\hat{\g}_{k}^{-1}\left(v\right)-\g_{k}^{-1}\left(v\right)\right| & \leq\frac{1}{\ul c}\left|\g_{k}\left(\hat{\g}_{k}^{-1}\left(v\right)\right)-\g_{k}\left(\g_{k}^{-1}\left(v\right)\right)\right|=\frac{1}{\ul c}\left|\g_{k}\left(\hat{\g}_{k}^{-1}\left(v\right)\right)-v\right|\\
 & =\frac{1}{\ul c}\left|\g_{k}\left(\hat{\g}_{k}^{-1}\left(v\right)\right)-\hat{\g}_{k}\left(\hat{\g}_{k}^{-1}\left(v\right)\right)\right|\leq\frac{1}{\ul c}\norm{\hat{\g}_{k}-\g_{k}}_{\infty}=O_{p}\left(a_{n}\right).
\end{align*}

Furthermore, observing that
\[
\ul c\left|\g_{k}^{-1}\left(v_{1}\right)-\g_{k}^{-1}\left(v_{2}\right)\right|\leq\left|\g_{k}\left(\g_{k}^{-1}\left(v_{1}\right)\right)-\g_{k}\left(\g_{k}^{-1}\left(v_{2}\right)\right)\right|=\left|v_{1}-v_{2}\right|
\]
we have by Assumption \ref{assu:StrongMono} and Lemma \ref{lem:bound_inv},
for $d_{i}=1$,
\begin{align}
\left|\hat{x}_{ik}-x_{ik}\right| & =\left|\hat{\g}_{j}^{-1}\left(\hat{\g}_{k}\left(x_{ik}\right)\right)-\g_{j}^{-1}\left(\g_{k}\left(x_{ik}\right)\right)\right|\nonumber \\
 & =\left|\hat{\g}_{j}^{-1}\left(\hat{\g}_{k}\left(x_{ik}\right)\right)-\g_{j}^{-1}\left(\hat{\g}_{k}\left(x_{ik}\right)\right)+\g_{j}^{-1}\left(\hat{\g}_{k}\left(x_{ik}\right)\right)-\g_{j}^{-1}\left(\g_{k}\left(x_{ik}\right)\right)\right|\nonumber \\
 & \leq\left|\hat{\g}_{j}^{-1}\left(\hat{\g}_{k}\left(x_{ik}\right)\right)-\g_{j}^{-1}\left(\hat{\g}_{k}\left(x_{ik}\right)\right)\right|+\left|\g_{j}^{-1}\left(\hat{\g}_{k}\left(x_{ik}\right)\right)-\g_{j}^{-1}\left(\g_{k}\left(x_{ik}\right)\right)\right|\nonumber \\
 & \leq\norm{\hat{\g}_{j}^{-1}-\g_{j}^{-1}}_{\infty}+\frac{1}{\ul c}\left|\hat{\g}_{k}\left(x_{ik}\right)-\g_{k}\left(x_{ik}\right)\right|\nonumber \\
 & \leq\norm{\hat{\g}_{j}^{-1}-\g_{j}^{-1}}_{\infty}+\frac{1}{\ul c}\norm{\hat{\g}_{k}-\g_{k}}_{\infty}\nonumber \\
 & =O_{p}\left(a_{n}\right).\label{eq:bound_h}
\end{align}
\end{proof}
\begin{lem}
\label{lem:FuncDeriv}Under Assumption \ref{assu:StrongMono}:
\begin{itemize}
\item[(i)] The pathwise derivative of $\gamma_{k}^{-1}$ w.r.t. $\gamma_{k}$
along $\tau_{k}\in\Gamma$ is given by
\begin{align*}
\nabla_{\gamma_{k}}\gamma_{k}^{-1}\left[\tau_{k}\right]:=\lim_{t\searrow0}\frac{\left(\gamma_{k}+t\tau_{k}\right)^{-1}\left(v\right)-\gamma_{k}^{-1}\left(v\right)}{t} & =-\frac{\tau_{k}\left(\gamma_{k}^{-1}\left(v\right)\right)}{\gamma_{k}^{^{\prime}}\left(\gamma_{k}^{-1}\left(v\right)\right)}.
\end{align*}
\item[(ii)] The pathwise derivative of $\gamma_{k}^{-1}\left(\gamma_{j}\left(\cdot\right)\right)$
w.r.t. $\gamma_{j}$ along $\tau_{j}\in\Gamma$ is given by
\begin{align*}
\nabla_{\gamma_{j}}\left(\gamma_{k}^{-1}\circ\gamma_{j}\right)\left[\tau_{j}\right]:= & \lim_{t\searrow0}\frac{\gamma_{k}^{-1}\left(\gamma_{j}\left(x\right)+t\tau_{j}\left(x\right)\right)-\gamma_{k}^{-1}\left(\gamma_{j}\left(x\right)\right)}{t}\\
= & \left(\gamma_{k}^{-1}\right)^{^{\prime}}\left(\gamma_{j}\left(x\right)\right)\tau_{j}\left(x\right)=\frac{1}{\gamma_{k}^{^{\prime}}\left(\gamma_{k}^{-1}\left(\gamma_{j}\left(x\right)\right)\right)}\tau_{j}\left(x\right).
\end{align*}
\item[(iii)] The second-order derivatives have bounded norms:
\begin{align*}
\nabla_{\gamma_{k}}^{2}\gamma_{k}^{-1}\left[\tau_{k}\right]\left[\tau_{k}\right] & \leq M\left\Vert \tau_{k}\right\Vert ^{2}\\
\nabla_{\gamma_{j}}^{2}\left(\gamma_{k}^{-1}\circ\gamma_{j}\right)\left[\tau_{j}\right]\left[\tau_{j}\right] & \leq M\left\Vert \tau_{k}\right\Vert ^{2}
\end{align*}
\end{itemize}
\end{lem}
\begin{proof}
(i) and (ii) follow immediately from the definition of pathwise derivatives.
See, e.g., Lemma 3.9.20 and 3.9.25 in \citet*{van1996weak} for reference.
For (iii),
\begin{align*}
\Dif_{\g_{k}}^{2}\g_{k}^{-1}\left[\tau_{k}\right]\left[\nu_{k}\right] & =\frac{\tau_{k}^{'}\left(\g_{k}^{-1}\right)}{\g_{k}^{'}\left(\g_{k}^{-1}\right)}\cd\frac{\nu_{k}\left(\g_{k}^{-1}\right)}{\g_{k}^{'}\left(\g_{k}^{-1}\right)}-\frac{\tau_{k}\left(\g_{k}^{-1}\right)}{\left[\g_{k}^{'}\left(\g_{k}^{-1}\right)\right]^{2}}\left[\g_{k}^{''}\left(\g_{k}^{-1}\right)+\frac{1}{\g_{k}^{'}\left(\g_{k}^{-1}\right)}\right]\nu_{k}\left(\g_{k}^{-1}\right)\\
 & \leq M\norm{\tau_{k}}\norm{\nu_{k}}
\end{align*}
since $\g_{k}^{'}\geq\ul c>0$ by Assumption \ref{assu:StrongMono}
and $\g^{''}$ and $\tau_{k}^{'}$ are uniformly bounded above by
Assumption \ref{assu:SP_RegCon}(i). Similarly for $\Dif_{\g_{j}}^{2}\left(\g_{k}^{-1}\circ\g_{j}\right)$.
\end{proof}
\begin{lem}
\label{lem:Influence} Writing $\gamma:=\left(\gamma_{1},\gamma_{2}\right)$,
the pathwise derivative of $\gamma_{k}^{-1}\circ\gamma_{j}$ w.r.t.
$\gamma$ along $\tau$ is given by
\begin{align*}
\nabla_{\gamma}\left(\gamma_{k}^{-1}\circ\gamma_{j}\right)\left[\tau\right]:= & \lim_{t\searrow0}\frac{\left(\gamma_{k}+t\tau_{k}\right)^{-1}\left(\gamma_{j}\left(x\right)+t\tau_{j}\left(x\right)\right)-\gamma_{k}^{-1}\left(\gamma_{j}\left(x\right)\right)}{t}\\
= & \frac{1}{\gamma_{k}^{^{\prime}}\left(\gamma_{k}^{-1}\left(\gamma_{j}\left(x\right)\right)\right)}\left[\tau_{j}\left(x\right)-\tau_{k}\left(\gamma_{k}^{-1}\left(\gamma_{j}\left(x\right)\right)\right)\right]
\end{align*}
\end{lem}
\begin{proof}
By Lemma \ref{lem:FuncDeriv},
\begin{align*}
 & \frac{1}{t}\left[\left(\g_{k}+t\tau_{k}\right)^{-1}\left(\g_{j}\left(x\right)+t\tau_{j}\left(x\right)\right)-\g_{k}^{-1}\left(\g_{j}\left(x\right)\right)\right]\\
=\  & \frac{1}{t}\left[\left(\g_{k}+t\tau_{k}\right)^{-1}\left(\g_{j}\left(x\right)+t\tau_{j}\left(x\right)\right)-\g_{k}^{-1}\left(\g_{j}\left(x\right)+t\tau_{j}\left(x\right)\right)\right]\\
 & +\frac{1}{t}\left[\g_{k}^{-1}\left(\g_{j}\left(x\right)+t\tau_{j}\left(x\right)\right)-\g_{k}^{-1}\left(\g_{j}\left(x\right)\right)\right]\\
\to\  & \Dif_{\g_{k}}\g_{k}^{-1}\left[\tau_{k}\right]\left(\g_{j}\left(x\right)\right)+\Dif_{\g_{j}}\left(\g_{k}^{-1}\circ\g_{j}\right)\left[\tau_{j}\right]\\
=\  & -\frac{\tau_{k}\left(\g_{k}^{-1}\left(\g_{j}\left(x\right)\right)\right)}{\g_{k}^{'}\left(\g_{k}^{-1}\left(\g_{j}\left(x\right)\right)\right)}+\frac{1}{\g_{k}^{'}\left(\g_{k}^{-1}\left(\g_{j}\left(x\right)\right)\right)}\tau_{j}\left(x\right)\\
=\  & \frac{1}{\g_{k}^{'}\left(\g_{k}^{-1}\left(\g_{j}\left(x\right)\right)\right)}\left(\tau_{j}\left(x\right)-\tau_{k}\left(\g_{k}^{-1}\left(\g_{j}\left(x\right)\right)\right)\right)
\end{align*}
\end{proof}

\subsection{Proof of Theorem \ref{thm:AsymDist}(i)}
\begin{proof}
We verify the conditions in Lemma 5.4 of \citet{newey1994asymptotic},
or equivalently, Theorems 8.11 of \citet*{newey1994large}.

Recall $w_{i}:=\left(y_{i},x_{i},z_{i},d_{i}\right)$, $\g:=\left(\g_{1},\g_{2}\right)$
and
\begin{align*}
g\left(w_{i},\hat{\a},\hat{\g}\right)= & \ol z_{i}\left(y_{i}-\hat{\a}_{0}-\left(x_{i1}\hat{\a}_{1}+\hat{\g}_{2}^{-1}\left(\hat{\g}_{1}\left(x_{i1}\right)\right)\hat{\a}_{2}\right)d_{i1}-\left(x_{i2}\hat{\a}_{2}+\hat{\g}_{1}^{-1}\left(\hat{\g}_{2}\left(x_{i2}\right)\right)\hat{\a}_{2}\right)d_{i2}\right)\\
= & \ol z_{i}\left(y_{i}-\hat{\a}_{0}-x_{ij}\hat{\a}_{j}-\hat{\g}_{k}^{-1}\left(\hat{\g}_{j}\left(x_{ij}\right)\right)\hat{\a}_{k}\right)\\
g\left(w_{i},\hat{\g}\right)= & \ol z_{i}\left(y_{i}-\a_{0}-\left(x_{i1}\a_{1}+\hat{\g}_{2}^{-1}\left(\hat{\g}_{1}\left(x_{i1}\right)\right)\a_{2}\right)d_{i1}-\left(x_{i2}\a_{2}+\hat{\g}_{1}^{-1}\left(\hat{\g}_{2}\left(x_{i2}\right)\right)\a_{2}\right)d_{i2}\right)\\
= & \ol z_{i}\left(y_{i}-\a_{0}-x_{ij}\a_{j}-\hat{\g}_{k}^{-1}\left(\hat{\g}_{j}\left(x_{ij}\right)\right)\a_{k}\right)\\
= & \ol z_{i}\left(u_{i}+\e_{i}+\left[x_{ik}-\hat{\g}_{k}^{-1}\left(\hat{\g}_{j}\left(x_{ij}\right)\right)\right]\a_{k}\right)
\end{align*}
Clearly, $\E\left[g\left(w_{i},\g\right)\right]=\E\left[\ol z_{i}\left(u_{i}+\e_{i}\right)\right]=0$
by Assumptions \ref{assu:IVz} and \ref{assu:eps_error}. Moreover,
$\frac{1}{N}\sum_{i=1}^{N}g\left(w_{i},\hat{\a},\hat{\g}\right)=0$
by the definition of $\hat{\a}$.

Now, define
\begin{align}
G\left(w_{i},\hat{\g}-\g\right) & :=\Dif_{\g}g\left(w_{i},\g\right)\left[\hat{\g}-\g\right]\nonumber \\
 & =-\a_{k}\ol z_{i}\Dif_{\g}\left(\g_{k}^{-1}\circ\g_{j}\right)\left[\hat{\g}-\g\right]\nonumber \\
 & =\frac{-\a_{k}\ol z_{i}}{\g_{k}^{'}\left(\g_{k}^{-1}\left(\g_{j}\left(x_{ij}\right)\right)\right)}\left[\left(\hat{\g}_{j}-\g_{j}\right)\left(x_{ij}\right)-\left(\hat{\g}_{k}-\g_{k}\right)\left(\g_{k}^{-1}\left(\g_{j}\left(x_{ij}\right)\right)\right)\right]\nonumber \\
 & =-\frac{\a_{k}\ol z_{i}}{\g_{k}^{'}\left(x_{ik}\right)}\left[\hat{\g}_{j}\left(x_{ij}\right)-\g_{j}\left(x_{ij}\right)-\hat{\g}_{k}\left(x_{ik}\right)+\g_{k}\left(x_{ik}\right)\right]\text{ since }\g_{k}^{-1}\left(\g_{j}\left(x_{ij}\right)\right)=x_{ik}\nonumber \\
 & =d_{i1}\ol z_{i}\left(-\frac{\a_{2}}{\g_{2}^{'}}\right)\left(1,-1\right)\left(\begin{array}{c}
\hat{\g}_{1}-\g_{1}\\
\hat{\g}_{2}-\g_{2}
\end{array}\right)+d_{i2}\ol z_{i}\left(-\frac{\a_{1}}{\g_{1}^{'}}\right)\left(-1,1\right)\left(\begin{array}{c}
\hat{\g}_{1}-\g_{1}\\
\hat{\g}_{2}-\g_{2}
\end{array}\right)\nonumber \\
 & =-\ol z_{i}\left(d_{i1}\frac{\a_{2}}{\g_{2}^{'}}-d_{i2}\frac{\a_{1}}{\g_{1}^{'}}\right)\left(1,-1\right)\left(\hat{\g}-\g\right)\label{eq:G_direct}
\end{align}
By Lemma \ref{lem:FuncDeriv}(iii) and Lemma \ref{lem:Influence},
we deduce
\[
\norm{g\left(w,\hat{\g}\right)-g\left(w,\g\right)-G\left(w,\hat{\g}-\g\right)}=O_{p}\left(\norm{\hat{\g}-\g}_{\infty}^{2}\right)=o_{p}\left(\frac{1}{\sqrt{N}}\right)
\]
given our assumption that $\norm{\hat{\g}-\g}_{\infty}=o_{p}\left(N^{-1/4}\right)$.

Next, the stochastic equicontinuity condition
\begin{equation}
\frac{1}{\sqrt{N}}\sum_{i=1}^{N}\left(G\left(w_{i},\hat{\g}-\g\right)-\int G\left(w_{i},\hat{\g}-\g\right)d\P\left(w_{i}\right)\right)=o_{p}\left(\frac{1}{\sqrt{N}}\right)\label{eq:stoc_ec}
\end{equation}
is guaranteed by Assumptions \ref{assu:StrongMono} and \ref{assu:SP_RegCon}.
Specifically, $\hat{\g}-\g$ belongs to a Donsker class of functions
by the smoothness assumption while $1/\g_{k}^{'}\left(x_{ik}\right)\leq1/\ul c$
guarantees that $G\left(z_{i},\cd\right)$ is square-integrable, so
that $G\left(z_{i},\cd\right)$ is also Donsker and thus \eqref{eq:stoc_ec}
holds.

Now, write $\zeta_{i}:=\left(x_{i},z_{i}\right)$ so that $w_{i}=\left(y_{i},\zeta_{i},d_{i}\right)$.
Then we have
\begin{align*}
 & \int G\left(w_{i},\hat{\g}-\g\right)\P w_{i}\\
= & \int-\ol z_{i}\left(d_{i1}\frac{\a_{2}}{\g_{2}^{'}}-d_{i2}\frac{\a_{1}}{\g_{1}^{'}}\right)\left(1,-1\right)\left(\hat{\g}-\g\right)d\P\left(\zeta_{i},d_{i}\right)\\
= & \int-\ol z_{i}\left(\left[\int d_{i1}d\P\left(\rest{d_{i}}\zeta_{i}\right)\right]\frac{\a_{2}}{\g_{2}^{'}}-\left[\int d_{i2}d\P\left(\rest{d_{i}}\zeta_{i}\right)\right]\frac{\a_{1}}{\g_{1}^{'}}\right)\left(1,-1\right)\left(\hat{\g}-\g\right)d\P\zeta_{i}\\
= & \int-\ol z_{i}\left(\l_{1}\left(\zeta_{i}\right)\frac{\a_{2}}{\g_{2}^{'}}-\l_{2}\left(\zeta_{i}\right)\frac{\a_{1}}{\g_{1}^{'}}\right)\left(1,-1\right)\left(\hat{\g}-\g\right)d\P\zeta_{i}
\end{align*}
By Proposition 4 of \citet*{newey1994asymptotic}, with
\[
\varphi\left(w_{i}\right):=-\left(\l_{1}\frac{\a_{2}\ol z_{i}}{\g_{2}^{'}}-\l_{2}\frac{\a_{1}\ol z_{i}}{\g_{1}^{'}}\right)\left(d_{i1}-d_{i2}\right)
\]
we have
\[
\ol z_{i}\left(\l_{1}\frac{\a_{2}}{\g_{2}^{'}}-\l_{2}\frac{\a_{1}}{\g_{1}^{'}}\right)\left(1,-1\right)\left(\begin{array}{c}
d_{i1}\left(y_{i}-\g_{1}\left(x_{i1}\right)\right)\\
d_{i2}\left(y_{i}-\g_{2}\left(x_{i2}\right)\right)
\end{array}\right)\equiv\varphi\left(w_{i}\right)\ol z_{i}\e_{i},
\]
and by Assumption \ref{assu:AsymLin}
\begin{align*}
\int G\left(w,\hat{\g}-\g\right)d\P\left(w\right) & =\frac{1}{N}\sum_{i=1}^{N}\varphi\left(w_{i}\right)\ol z_{i}\e_{i}+o_{p}\left(\frac{1}{\sqrt{N}}\right).
\end{align*}
Hence, Lemma 5.4 of \citet{newey1994asymptotic},
\begin{align*}
\frac{1}{\sqrt{N}}\sum_{i=1}^{N}g\left(w_{i},\hat{\g}\right) & =\frac{1}{\sqrt{N}}\sum_{i=1}^{N}\left[g\left(w_{i},\g\right)+\varphi\left(w_{i}\right)\ol z_{i}\e_{i}\right]+o_{p}\left(1\right)\dto\cN\left({\bf 0},\O\right),
\end{align*}
where
\begin{align*}
\O:= & \text{Var}\left[g\left(w_{i},\g\right)+\varphi\left(w_{i}\right)\ol z_{i}\e_{i}\right]\\
= & \E\left[\ol z_{i}\ol z_{i}^{'}\left(u_{i}+\left[1+\varphi\left(w_{i}\right)\right]\e_{i}\right)^{2}\right]=\E\left[\ol z_{i}\ol z_{i}^{'}\left(u_{i}^{2}+\left[1+\varphi\left(w_{i}\right)\right]^{2}\e_{i}^{2}\right)\right]
\end{align*}

Lastly, by Lemma \ref{lem:bound_inv}
\begin{align*}
\left|\frac{1}{n}\sum_{i=1}^{n}\ol z{}_{i}\left(\hat{x}_{i1}-x_{i1}\right)\right| & \leq\frac{1}{n}\sum_{i=1}^{n}\left|\ol z{}_{i}\right|\left|\hat{x}_{i1}-x_{i1}\right|\leq O_{p}\left(a_{n}\right)\cd\frac{1}{n}\sum_{i=1}^{n}\left|\ol z{}_{i}\right|=O_{p}\left(a_{n}\right)=o_{p}\left(1\right)
\end{align*}
and thus
\begin{align*}
\frac{1}{N}\sum_{i=1}^{N}\ol z_{i}\tilde{x}_{i}^{'} & =\E\left[\ol z_{i}\ol x_{i}^{'}\right]+\frac{1}{N}\sum_{i=1}^{N}\ol z_{i}\left(\tilde{x}_{i}-x_{i}\right)^{'}+\frac{1}{N}\sum_{i=1}^{N}\left(\ol z_{i}x_{i}^{'}-\E\left[\ol z_{i}x_{i}^{'}\right]\right)\\
 & =\E\left[\ol z_{i}\ol x_{i}^{'}\right]+O_{p}\left(a_{N}\right)+O_{p}\left(\frac{1}{\sqrt{N}}\right)\pto\Sigma_{zx}:=\E\left[\ol z_{i}\ol x_{i}^{'}\right].
\end{align*}
Hence,
\begin{align*}
\sqrt{N}\left(\hat{\a}-\a\right)=\left(\frac{1}{N}\sum_{i=1}^{N}\ol z_{i}\tilde{x}_{i}\right)^{-1}\frac{1}{\sqrt{N}}\sum_{i=1}^{N}g\left(w_{i},\hat{\g}\right) & \dto\cN\left({\bf 0},\Sigma_{zx}^{-1}\O\Sigma_{zx}^{'-1}\right).
\end{align*}
\end{proof}

\subsection{\label{subsec:Pf_asym_qbar}Proof of Theorem \ref{thm:AsymDist}(ii)}
\begin{proof}
We adapt the proof of Theorem \ref{thm:AsymDist}(i) above with
\begin{align*}
g^{*}\left(w,\hat{\a},\hat{\g}\right):= & \ol z_{i}\left(\hat{\g}_{j}\left(x_{ij}\right)-\hat{\a}_{0}-\hat{\a}_{j}x_{ij}-\hat{\a}_{k}\hat{\g}_{k}^{-1}\left(\hat{\g}_{j}\left(x_{ij}\right)\right)\right),\\
g^{*}\left(w,\hat{\g}\right):= & \ol z_{i}\left(\hat{\g}_{j}\left(x_{ij}\right)-\a_{0}-\a_{j}x_{ij}-\a_{k}\hat{\g}_{k}^{-1}\left(\hat{\g}_{j}\left(x_{ij}\right)\right)\right).
\end{align*}
with $\E\left[g^{*}\left(w_{i},\g\right)\right]=\E\left[\ol z_{i}\left(\g_{j}\left(x_{ij}\right)-\a_{0}-\a_{j}x_{ij}-\a_{k}\g_{k}^{-1}\left(\g_{j}\left(x_{ij}\right)\right)\right)\right]=\E\left[\ol z_{i}u_{i}\right]={\bf 0}$
and $\frac{1}{N}\sum_{i=1}^{N}g\left(z,\hat{\a}^{*},\hat{\g}\right)={\bf 0}$.

By the chain rule,
\begin{align*}
G^{*}\left(w_{i},\tau\right):= & \Dif_{\g}g^{*}\left(w_{i},\g\right)\left[\hat{\g}-\g\right]\\
= & \ol z_{i}\left(\left[\hat{\g}_{j}\left(x_{ij}\right)-\g_{j}\left(x_{ij}\right)\right]-\a_{k}\Dif_{\g}\left(\g_{k}^{-1}\circ\g_{j}\right)\left[\hat{\g}-\g\right]\right)\\
= & \ol z_{i}\left(1-\frac{\a_{k}}{\g_{k}^{'}\left(x_{ik}\right)}\right)\left[\hat{\g}_{j}\left(x_{ij}\right)-\g_{j}\left(x_{ij}\right)\right]-\ol z_{i}\frac{\a_{k}}{\g_{k}^{'}\left(x_{ik}\right)}\left[\hat{\g}_{k}\left(x_{ik}\right)-\g_{k}\left(x_{ik}\right)\right]\\
= & \ol z_{i}\left[d_{i1}\left(1-\frac{\a_{2}}{\g_{2}^{'}},-\frac{\a_{2}}{\g_{2}^{'}}\right)+d_{i2}\left(-\frac{\a_{1}}{\g_{1}^{'}},1-\frac{\a_{1}}{\g_{1}^{'}}\right)\right]\left(\hat{\g}-\g\right)
\end{align*}
and
\begin{align*}
\int G\left(w_{i},\hat{\g}-\g\right)\P w_{i}= & \int\ol z_{i}\left(\l_{1}\left(1-\frac{\a_{2}}{\g_{2}^{'}}\right)+\l_{2}\frac{\a_{1}}{\g_{1}^{'}},\ \l_{1}\frac{\a_{2}}{\g_{2}^{'}}+\l_{2}\left(1-\frac{\a_{1}}{\g_{1}^{'}}\right)\right)\left(\hat{\g}-\g\right)d\P\zeta_{i}
\end{align*}
By Proposition 4 of \citet*{newey1994asymptotic}, with
\[
\varphi^{*}\left(w_{i}\right):=-\left(\l_{1}\left(1-\frac{\a_{2}}{\g_{2}^{'}}\right)+\l_{2}\frac{\a_{1}}{\g_{1}^{'}}\right)d_{i1}+\left(\l_{1}\frac{\a_{2}}{\g_{2}^{'}}+\l_{2}\left(1-\frac{\a_{1}}{\g_{1}^{'}}\right)\right)d_{i2}
\]
we have
\[
\ol z_{i}\left(\l_{1}\left(1-\frac{\a_{2}}{\g_{2}^{'}}\right)+\l_{2}\frac{\a_{1}}{\g_{1}^{'}},\ \l_{1}\frac{\a_{2}}{\g_{2}^{'}}+\l_{2}\left(1-\frac{\a_{1}}{\g_{1}^{'}}\right)\right)\left(\begin{array}{c}
d_{i1}\left(y_{i}-\g_{1}\left(x_{i1}\right)\right)\\
d_{i2}\left(y_{i}-\g_{2}\left(x_{i2}\right)\right)
\end{array}\right)\equiv\varphi^{*}\left(w_{i}\right)\ol z_{i}\e_{i},
\]
and by Assumption \ref{assu:AsymLin}
\begin{align*}
\int G\left(w,\hat{\g}-\g\right)d\P\left(w\right) & =\frac{1}{N}\sum_{i=1}^{N}\varphi^{*}\left(w_{i}\right)\ol z_{i}\e_{i}+o_{p}\left(\frac{1}{\sqrt{N}}\right).
\end{align*}
Hence, we have
\begin{align*}
\frac{1}{\sqrt{N}}\sum_{i=1}^{N}g^{*}\left(w_{i},\hat{\g}\right) & =\frac{1}{\sqrt{N}}\sum_{i=1}^{N}\left[g^{*}\left(w_{i},\g\right)+\varphi^{*}\left(w_{i}\right)\ol z_{i}\right]+o_{p}\left(1\right)\dto\cN\left({\bf 0},\O^{*}\right),
\end{align*}
where
\begin{align*}
\O & :=\text{Var}\left[g^{*}\left(w_{i},\g\right)+\d^{*}\left(z_{i}\right)\right]=\E\left[\ol z_{i}\ol z_{i}^{'}\left(u_{i}^{2}+\varphi^{*}\left(w_{i}\right)^{2}\e_{i}^{2}\right)\right],
\end{align*}
giving
\begin{align*}
\sqrt{N}\left(\hat{\a}-\a\right)=\left(\frac{1}{N}\sum_{i=1}^{N}\ol z_{i}\tilde{x}_{i}\right)^{-1}\frac{1}{\sqrt{N}}\sum_{i=1}^{N}g^{*}\left(w_{i},\hat{\g}\right) & \dto\cN\left({\bf 0},\Sigma_{zx}^{-1}\O^{*}\Sigma_{zx}^{'-1}\right).
\end{align*}
\end{proof}

\subsection{\label{subsec:pf_StarBetter}Proof of Theorem \ref{thm:StarBetter}}
\begin{proof}
By \eqref{eq:gamma1_detail}, we have
\[
\frac{\p}{\p c}\g_{j}\left(c;z\right)=\a_{j}+\a_{k}x_{k}^{'}\frac{1}{x_{j}^{'}}+\frac{1}{x_{j}^{'}}>\a_{j},
\]
and thus $0<\a_{j}/\g_{j}^{'}<1$, which implies
\[
\l_{1}\left(1-\frac{\a_{2}}{\g_{2}^{'}}\right)+\l_{2}\frac{\a_{1}}{\g_{1}^{'}}>0,\quad\l_{2}\left(1-\frac{\a_{1}}{\g_{1}^{'}}\right)+\l_{1}\frac{\a_{2}}{\g_{2}^{'}}>0.
\]
Hence,
\begin{align*}
\varphi^{*} & =\left(\l_{1}\left(1-\frac{\a_{2}}{\g_{2}^{'}}\right)+\l_{2}\frac{\a_{1}}{\g_{1}^{'}}\right)d_{i1}+\left(\l_{2}\left(1-\frac{\a_{1}}{\g_{1}^{'}}\right)+\l_{1}\frac{\a_{2}}{\g_{2}^{'}}\right)d_{i2}>0\\
1+\varphi & =1-\left(\frac{\a_{2}}{\g_{2}^{'}}\l_{1}-\frac{\a_{1}}{\g_{1}^{'}}\l_{2}\right)\left(d_{i1}-d_{i2}\right)\\
 & =\left(1-\l_{1}\frac{\a_{2}}{\g_{2}^{'}}+\l_{2}\frac{\a_{1}}{\g_{1}^{'}}\right)d_{i1}+\left(1-\l_{2}\frac{\a_{1}}{\g_{1}^{'}}+\l_{1}\frac{\a_{2}}{\g_{2}^{'}}\right)d_{i2}\\
 & =\varphi^{*}+\left(1-\l_{1}\right)d_{i1}+\left(1-\l_{2}\right)d_{i2}\\
 & >\varphi^{*}>0.
\end{align*}
Hence, $\left(1+\varphi\right)^{2}>\varphi^{*2}>0$ and
\begin{align*}
\O-\O^{*} & =\E\left[\ol z_{i}\ol z_{i}^{'}\left[\left(1-\varphi\left(x_{i},d_{i}\right)\right)^{2}-\varphi^{*}\left(x_{i},d_{i}\right)^{2}\right]\e_{i}^{2}\right]
\end{align*}
is positive definite. Therefore, $\Sigma-\Sigma^{*}$ is also positive
definite and $\hat{\a}^{*}$ is asymptotically more efficient than
$\hat{\a}$.
\end{proof}

\subsection{Proof of Theorem \ref{thm:Kernel}}
\begin{proof}
Assumption \ref{assu:KernelConds}(i) guarantees that $N_{1}\sim N_{2}\sim N$
so that
\[
\norm{\hat{\g}_{1}-\g_{1}}_{\infty}\sim\norm{\hat{\g}_{2}-\g_{2}}_{\infty}=O_{p}\left(a_{N}\right)
\]
where, by Assumption \ref{assu:KernelConds}(ii)-(v) and Theorem 8
of \citet*{hansen2008uniform},
\[
a_{N}=b^{p}+\frac{\sqrt{\log N}}{\sqrt{Nb^{3}}}.
\]
With $b$ chosen according to Assumption \ref{assu:KernelConds}(vi)
so that $\frac{\sqrt{\log N}}{\sqrt{Nb^{3}}}=o\left(N^{-\frac{1}{4}}\right)$
and $\sqrt{N}b^{p}\to0$, implying that
\[
a_{N}=o\left(N^{-\frac{1}{2}}\right)+o\left(N^{-\frac{1}{4}}\right)=o\left(N^{-\frac{1}{4}}\right),
\]
verifying Assumption \ref{assu:SP_RegCon}(ii). Assumption \ref{assu:AsymLin}
(and consequently Theorem \ref{thm:Kernel}) follows from Theorem
8.11 of \citet*{newey1994large}.
\end{proof}

\section{{\normalsize{}\label{sec:MCPartialWage}}Monte Carlo with Partially
Latent Wages}

In this section, we consider the case in which the wage for type $j$
is observed only when we observe the input for type $j$. This data
structure can typically be observed when we have individual-level
survey data where each individual reports his/her own inputs and wages.
Specifically, we have that:
\begin{align}
(z_{i1},z_{i2}) & =\begin{cases}
(z_{i1},missing)\qquad\text{if }x_{i1}\text{ is observed}\\
(missing,z_{i2})\qquad\text{if }x_{i2}\text{ is observed}
\end{cases}
\end{align}
Since we need to impute missing wages, we assume that wages differ
across, but not within local labor markets. Let $m(i)$ denote the
local market that firm $i$ is active in. If we observe enough firms
in a local market than we can treat both wages as observed. Below
we consider the case in which we only observe one firm per local labor
market, but wages can be expressed as functions of some demand shifters
$D_{m}\in\mathbb{R}^{2}$ for the local labor market $m$ and a random
error $\eta_{i}$ which is assumed to be independent from the demand
shifters. Note that this specification allows for correlation between
$z_{1m(i)}$ and $z_{2m(i)}$ through $D_{m}$. Specifically, we simulate
wages as follows:
\begin{align}
z_{i1}= & \kappa_{1}D_{m}+\eta_{i1}\\
z_{i2}= & \kappa_{2}D_{m}+\eta_{i2}
\end{align}

The demand shifters are drawn randomly from a bivariate normal distribution
with a mean of $\mu_{D}=(2.4,2.1)'$ and a variance of $\Sigma_{D}=(0.05,0;0,0.02)$.
And the parameter values are set to be $\kappa_{1}=(1.1,0.3)$ and
$\kappa_{2}=(0.1,0.9)$. All other parameter values vary across the
specifications reported in Table \ref{ParamSpec}.

To impute the missing wages, we regress the observed wages $(z_{i1},z_{i2})$
on the demand shifters ($D_{m}$). Using estimated parameters from
the regression, we then impute the missing local labor market wages.

\begin{table}[htbp]
\caption{Monte Carlo: Small Markets with Partially Latent Wages }
\label{ImputeWage}\centering %
\begin{tabular}{cccccccc}
\toprule
 & Number of  & Number of  &  & \multicolumn{2}{c}{Standard SLS} & \multicolumn{2}{c}{Matched TSLS}\tabularnewline
\cmidrule(lr){5-6} \cmidrule(lr){7-8} Param  & markets  & firms  & Spec  & Bias  & RMSE  & Bias  & RMSE \tabularnewline
\midrule
$\alpha_{0}$  & 500  & 1  & 1  & -0.002  & 0.002  & -0.002  & 0.002 \tabularnewline
$\alpha_{0}$  & 500  & 1  & 2  & -0.010  & 0.006  & -0.010  & 0.006 \tabularnewline
$\alpha_{0}$  & 500  & 1  & 3  & -0.002  & 0.002  & -0.002  & 0.002 \tabularnewline
$\alpha_{0}$  & 500  & 1  & 4  & -0.001  & 0.002  & -0.000  & 0.003 \tabularnewline
\midrule
$\alpha_{1}$  & 500  & 1  & 1  & 0.001  & 0.014  & 0.003  & 0.016 \tabularnewline
$\alpha_{1}$  & 500  & 1  & 2  & 0.007  & 0.049  & 0.010  & 0.056 \tabularnewline
$\alpha_{1}$  & 500  & 1  & 3  & 0.001  & 0.014  & 0.003  & 0.016 \tabularnewline
$\alpha_{1}$  & 500  & 1  & 4  & 0.001  & 0.022  & 0.009  & 0.033 \tabularnewline
\midrule
$\alpha_{2}$  & 500  & 1  & 1  & -0.007  & 0.019  & -0.007  & 0.021 \tabularnewline
$\alpha_{2}$  & 500  & 1  & 2  & -0.028  & 0.067  & -0.030  & 0.077 \tabularnewline
$\alpha_{2}$  & 500  & 1  & 3  & -0.007  & 0.019  & -0.007  & 0.021 \tabularnewline
$\alpha_{2}$  & 500  & 1  & 4  & -0.006  & 0.030  & -0.013  & 0.043 \tabularnewline
\bottomrule
\end{tabular}
\end{table}

Table \ref{ImputeWage} summarizes the performance of our new estimator
together with TSLS estimator. Even if we have a relatively large variance
of the measurement errors, such as in Specification 4, our new estimator
performs reasonably well.

\section{Applications outside of Industrial Organization}

In this appendix we discuss that the partially latent covariate problem can also arise
in the study of intergenerational transmission of human capital, wealth
or attitudes. It is rather common in many data sets that we do not sample all relevant
household members. We then illustrate these ideas and estimate the achievement
function of children with and without divorced parents.

\subsection{Discussion}

There are a number of surveys that have a sampling design that gives
rise to partially latent covariates. A prototypical survey is the
Panel Study of Income Dynamics (PSID). The survey starts with an initial
set of respondents that is representative of the U.S. population in
1968. The initial survey collects the answers to a large number of
questions about the respondent's background and financial circumstances,
including family information (married or single, number of children,
etc.) Importantly, all children born to a PSID respondent join the
PSID sample and are included in future waves. As a result, the PSID
data set has substantial information about individuals with the ``PSID
gene'', that is, individuals who are direct descendants of the initial
cohort. These individuals typically marry individuals who lack the
PSID gene, and consequently, much more is known about the individual
with the gene than about the individual without the gene. Hence, the
PSID data set has, by construction, a large number of couples for
which many variables are partially latent. For example, if we wanted
to know whether parental teenage employment affects offspring education
attainment, we have the teenage employment for one parent but not
the other (unless, of course, the matched couple each had the gene.)
The inclusion of the CDS sample amplifies the effect: for the youngest
children, we have information about their grandparents with the PSID
seed, but nothing about the other set of grandparents.

More generally, the partially latent covariate problem arises in studies
of intergenerational outcomes, such as the transmission of wealth
or attitudes. For example, researchers have been interested in inter
vivos gifts. It is common for parents, while still alive, to give
money to their children, often to help with a down payment on a house
or to reduce taxes the parents will pay. When a couple makes a gift
to their married child, however, they risk that the child divorces
and a portion of the gift will accrue to the child's spouse. The concern
is real since approximately 40\% of marriages in the U.S. end in divorce.
A natural question is how well parents can predict how long a child's
marriage will last at the time they contemplate making a gift. One
could address this question with a data set that includes inter vivos
gifts from parents to married children and, in addition, how long
the child's marriage survives. Such data sets exist, for example the
PSID, which documents these for family lines that stretch over a half
century. In this application, we know the inter vivos gifts to the
couple from the parents of the PSID gene child but not inter vivos
gifts to the couple from the spouse's parents. It is straightforward
to write down a non-cooperative model of intergenerational transfer,
where the transfers of both sets of parents are monotonically increasing
in the estimated probability that the marriage survives. Hence, the
model satisfies our monotonicity assumption.

Similar data structures arise in the National Longitudinal Survey
of Youth, American Community Survey, Health and Retirement Study, and
the American Time Use Survey. The bottom line is that for a large
set of questions relating to family economics and the intergenerational
transmission of wealth, human capital, and attitudes, the data one
might use is often survey data that often very naturally has partially
latent variables. The methods developed in this paper are, therefore,
central to improving our understanding intergenerational mobility
and intergenerational poverty.\footnote{The NLSY, HRS and ATS contain more information about the initial respondents
than their spouses. The ACS collects more data about the person filling
out the questionnaire than the spouse.} This is not to say that the techniques in our paper will necessarily
solve all such latent variable issues. However, we conjecture that
the methods developed in this paper should provide useful insights
into the study of a variety of questions related to intergenerational
linkages.

\subsection{An Application: Achievement Functions}

To illustrate these ideas, we consider, the Child Development Supplement
of the PSID to study the transmission of human capital from parents
to children. If a child grows up in a traditional family, it is likely
that we observe inputs from both mother and father. However, traditional
family arrangements have become less common during the past decades
and thus the focus has shifted on evaluating the impact of growing
up in less traditional environments. If a child grows up in a non-traditional
family one of the two parents' inputs is often missing.

Our data is based on the four available waves of the Child Development
Supplement (CDS). These are the cohorts interviewed in 1997, 2002,
2007, and 2014.\footnote{The CDS 1997 cohort consists of up to 12-year-old children and follows
them for 3 waves (1997, 2001, 2007). The CDS 2014 cohort consists
of children that were up to 17 years old in 2013.} For these children, we have detailed time usage information of their
parents on two days, each of which is randomly selected among weekdays
and weekends, respectively. Based on this time diary information we
can construct time inputs for mothers and fathers.\footnote{We exclude families with stepmother and stepfather from our sample.}
The CDS can be linked to the original PSID survey using the family
ID. Hence, we have detailed parental information such as education
level, household income, and the number of children.

The CDS collects multiple measures of child development including
both cognitive and non-cognitive skills. We focus on two important
cognitive tests. First, we study the passage comprehension test which
assesses reading comprehension and vocabulary among children aged
between 6 and 17. Second, we analyze the applied problems test which
assesses mathematics reasoning, achievement, and knowledge for children
aged between 6 and 17.\footnote{We also analyzed the letter word test which assesses symbolic learning
and reading identification skills. There are also two non-cognitive
measures. The externalizing behavioral problem index measures disruptive,
aggressive, or destructive behavior. The internalizing behavioral
problem index measures expressions of withdrawn, sad, fearful, or
anxious feelings.}

Here we assume that a child's achievement $y_{i}$ is a function of
the mother's and the father's time inputs, denoted by $x_{im}$ and
$x_{if}$. Again, we consider a log-linear Cobb-Douglas specification
given by
\begin{eqnarray}
y_{i} & = & \alpha_{i}\;+\;\alpha_{m}\;x_{im}\;+\;\alpha_{f}\;x_{if}\;+\;u_{i}\label{ed_prod}
\end{eqnarray}
where heterogeneity in the intercept is given by:
\begin{eqnarray}
\alpha_{i} & = & x_{i}^{\prime}\;\alpha_{0}\label{ed_het}
\end{eqnarray}
Hence, we assume that the baseline marginal product of inputs $\alpha_{i}$ varies
with family characteristics, such as family income. As before, we
can estimate the education production function using TSLS with wages
as instruments for inputs as well as our ``matched'' TSLS estimator
if some inputs are partially latent.

We begin by estimating an achievement function using the subsample
of children who live in married households. Hence, we observe the
mother's and the father's inputs in the data set. We observe 3,236
children with complete inputs and applied problem scores as well as
2,789 children with complete inputs and reading comprehension scores.
Table \ref{ed_summarystat} provides descriptive statistics of the
main variables in our sample.

\begin{table}[htbp]
\caption{Summary Statistics of CDS Sample}
\label{ed_summarystat}\centering{}%
\begin{tabular}{lcc}
\toprule
 & Married Sample  & Divorced Sample \tabularnewline
\midrule
Applied Problem Score (Standardized)  & 107.58  & 101.28 \tabularnewline
 & (16.63)  & (16.92) \tabularnewline
Passage Comprehension Score (Standardized)  & 105.89  & 99.48 \tabularnewline
 & (14.77)  & (14.49) \tabularnewline
\midrule
Mother's Time Input  & 20.77  & 15.18 \tabularnewline
 & (14.32)  & (14.06) \tabularnewline
Father's Time Input  & 13.87  & 4.34 \tabularnewline
 & (11.96)  & (13.81) \tabularnewline
\midrule
Total Number of Child In Family  & 2.17  & 2.1 \tabularnewline
 & (0.9)  & (0.9) \tabularnewline
Child's Age At Interview  & 9.68  & 11.37 \tabularnewline
 & (4.74)  & (4.44) \tabularnewline
Total Household Labor Income (in 2011 Dollar)  & 68941  & 24158 \tabularnewline
 & (55732)  & (28616) \tabularnewline
Mother's Age  & 37.05  & 37.3 \tabularnewline
 & (7.27)  & (6.85) \tabularnewline
Father's Age  & 39.1  & 38.81 \tabularnewline
 & (7.7)  & (8.8) \tabularnewline
Mother's Years of Education  & 13.51  & 12.92 \tabularnewline
 & (2.57)  & (1.97) \tabularnewline
Father's Years of Education  & 13.38  & 12.97 \tabularnewline
 & (3.21)  & (1.9) \tabularnewline
Living With Mother  & -  & 0.88 \tabularnewline
\bottomrule
 &  & \tabularnewline
\end{tabular}
\end{table}

We can estimate the model using the traditional TSLS estimator. We
compare these estimates with our matched TSLS which is based on a
sample in which we randomly omit one of the two inputs. This exercise
allows us to compare the performance of both estimators when there
is no latent input problem. We restrict our attention to married couples
with both spouses living together. We exclude families with more than
5 children. As instruments for time inputs we use education, employment
status, hourly wage, age of children. To preserve the representativeness
of our sample, we use the child-level survey weight for all analyses.
Household labor income is measured in 10,000 dollars. Table \ref{ed_prod_mar}
summarizes our findings.

\begin{table}[htbp]
\caption{Education Production Function: Married Sample}
\label{ed_prod_mar}\centering{} %
\begin{tabular}{lcccc}
\toprule
 & \multicolumn{2}{c}{Applied Problems} & \multicolumn{2}{c}{Passage Comprehension}\tabularnewline
 & TSLS  & matched TSLS  & TSLS  & matched TSLS \tabularnewline
\midrule
Mom Hour  & 0.016  & 0.027  & 0.100  & 0.098 \tabularnewline
 & (0.008)  & (0.002)  & (0.012)  & (0.033) \tabularnewline
Dad Hour  & 0.032  & 0.021  & 0.017  & 0.006 \tabularnewline
 & (0.007)  & (0.007)  & (0.009)  & (0.040) \tabularnewline
\midrule
Num Child = 2  & $-$0.011  & 0.034  & $-$0.051  & $-$0.097 \tabularnewline
 & (0.008)  & (0.020)  & (0.013)  & (0.150) \tabularnewline
Num Child = 3+  & 0.008  & 0.077  & $-$0.030  & $-$0.059 \tabularnewline
 & (0.009)  & (0.026)  & (0.014)  & (0.152) \tabularnewline
Household Labor Inc  & 0.008  & 0.006  & 0.010  & 0.009 \tabularnewline
 & (0.001)  & (0.002)  & (0.001)  & (0.017) \tabularnewline
Constant  & 4.510  & 4.484  & 4.321  & 4.380 \tabularnewline
 & (0.017)  & (0.026)  & (0.026)  & (0.223) \tabularnewline
\midrule
Nobs  & 3,236  & 3,236  & 2,789  & 2,789 \tabularnewline
First-stage F for $x_{m}$  & 61.997  & 127.295  & 41.812  & 58.530 \tabularnewline
First-stage F for $x_{f}$  & 62.636  & 117.966 &  & \tabularnewline
\bottomrule
\end{tabular}
\end{table}

Overall, our empirical findings are reasonable. We find that investments
in child quality decrease with the number of children in the family
and increase with household income, as expected. Both parental time
inputs are positive and typically statistically significant and economically
meaningful. Comparing the TSLS with our matched TSLS estimator, we
find that the results are remarkably similar, especially for the passage
comprehension test. The results for the applied problem test are also
encouraging although the differences in the estimates are slightly
larger. Qualitatively, we reach the same conclusions with both estimators.
We thus conclude that our matched TSLS performs well in this sample.

Next, we consider the subsample that consists of households that self-reported
to be either divorced or separated. We exclude single households for
obvious reasons. In all households in this sample one of the parents
is not living in the child's household. We typically do not observe
time inputs for these divorced parents. For the applied problem (passage
comprehension) score we observe 785 (723) children with the mother's
input. There are 103 (92) observations where we have the father's
input, which we use for imputation purposes. Again, we use reported hourly wages for the observed parent and imputed hourly wages for the parent who are not in the household. Also, we use a few other variables as instruments such as education levels, years of work experience, and the employment status interacted with the wages of both parents. We acknowledge that instruments such as wages, employment status or work experience may be problematic since they may also be correlated with unobserved variables affecting child care decisions. We impute these instruments from observed characteristics of the household for the spouse whose input is latent in the data.\footnote{Note that the standard TSLS is no longer feasible in this subsample because of the latent variable problem.}  Table \ref{ed_prod_div} summarizes
our findings.

\begin{table}
\caption{Education Production Function: Divorced Sample}
\label{ed_prod_div}\centering{}%
\begin{tabular}{lcc}
\toprule
 & \multicolumn{1}{c}{Applied Problems} & \multicolumn{1}{c}{Passage Comprehension}\tabularnewline
 & matched TSLS  & matched TSLS \tabularnewline
\midrule
Mom Hour  & 0.050  & 0.037 \tabularnewline
 & (0.028)  & (0.015) \tabularnewline
Dad Hour  & 0.010  & 0.001 \tabularnewline
 & (0.013)  & (0.003) \tabularnewline
\midrule
Num Child = 2  & 0.051  & 0.019 \tabularnewline
 & (0.055)  & (0.039) \tabularnewline
Num Child = 3+  & 0.002  & $-$0.015 \tabularnewline
 & (0.056)  & (0.066) \tabularnewline
Household Labor Inc  & $-$0.013  & $-$0.006 \tabularnewline
 & (0.016)  & (0.004) \tabularnewline
Constant  & 4.548  & 4.529 \tabularnewline
 & (0.078)  & (0.061) \tabularnewline
\midrule
Nobs  & 785  & 723 \tabularnewline
First-stage F for $x_{m}$  & 40.532  & 35.264 \tabularnewline
First-stage F for $x_{f}$ &  & \tabularnewline
\bottomrule
\end{tabular}
\end{table}

Table \ref{ed_prod_div} shows that the time inputs for mothers are
positive, statistically significant, and economically meaningful.
Moreover, the point estimates for the applied problem test are similar
to the ones we obtained for the married sample reported in Table \ref{ed_prod_mar}.
The main difference is that mother's time inputs are slightly less
productive for children from divorced families, and father's time
inputs are not statistically different from zero. In summary, our
estimator works well in this application and yields plausible and
accurate point estimates for most coefficients of interest. Most importantly,
we find that the inputs of divorced fathers into the skill formation
function of their children are negligible.

\end{document}